\documentclass[final,5p,times]{elsarticle}
\usepackage{amssymb,amsmath,latexsym,amsthm,amsfonts,mathtools}

\usepackage{mathrsfs}
\usepackage{multirow}
\usepackage{graphicx}
\usepackage{textcomp,siunitx}
\usepackage{float}
\usepackage{subcaption}
\usepackage{booktabs}
\usepackage{enumitem}
\usepackage{bm}
\usepackage{hyperref}
\usepackage{comment}
\hypersetup{colorlinks, breaklinks, citecolor=blue, linkcolor=red, urlcolor=blue}
\usepackage[nameinlink,noabbrev,sort&compress]{cleveref}
\crefname{assumption}{Assumption}{Assumptions}
\theoremstyle{plain}
\newtheorem{theorem}{Theorem}

\newtheorem{corollary}{Corollary}

\newtheorem{assumption}{Assumption}

\theoremstyle{remark}

\theoremstyle{definition}
\newtheorem{definition}{Definition}%[section]

\newtheorem{problem}{Problem}
\newcommand{\sign}{\mathrm{sign}}

\usepackage{accents}

\usepackage{lineno}
\modulolinenumbers[5]

\journal{Aerospace Science and Technology}

\biboptions{numbers,sort&compress}

\begin{document}
	\begin{frontmatter}
		\title{Path-Following Guidance for Unmanned Aerial Vehicle with Bounded Lateral Acceleration}
		\author[srk]{Vinay Kathiriya}
		\ead{23b0021@iitb.ac.in}  
		\author[srk]{Saurabh Kumar}
		\ead{saurabh.k@aero.iitb.ac.in}
		\author[srk]{Shashi Ranjan Kumar}
		\ead{srk@aero.iitb.ac.in}
		\address[srk]{Intelligent Systems and Control Lab, Department of Aerospace Engineering, Indian Institute of Technology Bombay, Powai, Mumbai- 400076, India}
		\begin{abstract}
			This paper addresses the three-dimensional path-following guidance problem for unmanned aerial vehicles under explicit actuator constraints. Unlike conventional approaches that assume unbounded control inputs or handle saturation heuristically, the proposed method incorporates bounded lateral acceleration directly into the guidance design. A nonlinear guidance framework is developed employing a nested saturation-based control technique. The proposed guidance strategy guarantees bounded control inputs while ensuring exponential convergence of cross-track errors to zero. The formulation is applicable to general smooth paths and is systematically extended from planar to three-dimensional scenarios using a path-tangent coordinate framework. Rigorous stability analysis based on Lyapunov theory establishes convergence and feasibility properties of the closed-loop system. Numerical simulations on representative paths, including straight-line, circular, and sinusoidal paths, demonstrate that the proposed method achieves superior tracking performance, reduced control effort, and robustness against disturbances compared to existing guidance laws. The simplicity of the design and its compatibility with practical actuator limits make it suitable for real-world UAV applications.
		\end{abstract}
		
		\begin{keyword}
			Motion control, autonomous vehicle, path-following, nonlinear guidance, unmanned aerial systems, bounded control. 
		\end{keyword}
	\end{frontmatter}
	
	\section{Introduction}\label{sec:intro}
	The use of unmanned aerial vehicles (UAVs) has expanded rapidly in both military and civilian domains. In defense applications, UAVs are employed for surveillance, reconnaissance, target interception, and logistics support. Civilian applications include mapping, surveying, infrastructure inspection, search-and-rescue operations, precision agriculture, traffic monitoring, and medical supply delivery to remote regions. With such diverse and safety-critical applications, it is essential to develop robust guidance systems capable of performing autonomous tasks such as path following, target tracking, obstacle avoidance, and trajectory tracking. To achieve reliable performance, UAVs should exhibit high accuracy, agility, and resilience against complex coupled dynamics, while simultaneously respecting practical constraints such as actuator saturation, strict path requirements, and external disturbances. In particular, actuator limitations, such as control surface saturation, may destabilize the system and degrade tracking accuracy. Motivated by these challenges, this work focuses on path-following under bounded-input constraints.  
	
	Early UAV guidance strategies were adapted from those used in interceptor and marine systems. The pure pursuit (PP) approach, initially developed for target interception, was later applied to UAVs, where the vehicle follows a moving point on the desired path, which was used to generate a smooth trajectory \cite{yamasaki2009robust, coulter1992implementation, chen2019trajectory}. Later, this idea was further extended by introducing the concept of a virtual target point that moves along the reference path, known as the carrot-chasing algorithm \cite{jin2020path,sujit2013evaluation}, where the UAV continually pursues a virtual target placed at a fixed look-ahead distance ahead of the UAV's closest point on the path. The authors in \cite{rysdyk2003uav} employed a geometric approach, utilizing a line-of-sight (LOS) guidance law that steers the UAV's velocity vector toward an aim point on the path. Unlike the classical LOS method, which used heading commands, the methods in \cite{park2007performance} used lateral acceleration commands to steer the UAV to a virtual target on a look-ahead circle of radius $L_1$. However, the performance deteriorates when the UAV is operated outside the look-ahead circle. The $L_2$ guidance law \cite{curry2013l+} addressed this limitation by adaptively adjusting the look-ahead distance to reduce the sensitivity to ground speed and large cross-track error. A hybrid method, pure pursuit and line-of-sight (PLOS), that combines PP and LOS principles was introduced in \cite{jin2020path} and achieves better performance than the individual methods.
	
	Another approach, known as the vector field (VF) \cite{nelson2007vector, sun2022stability}, has also been utilized to design path-following guidance strategies. This method generates ground-track vector fields whose gradients ensure asymptotic convergence of cross-track error and heading angle error, while maintaining robustness in constant-wind conditions. However, VF methods do not explicitly handle input constraints, leading to infeasible commands during aggressive maneuvers. To improve the performance of the UAV, various nonlinear guidance strategies based on different techniques, such as backstepping \cite{mohd2015enhanced,asl2017adaptive}, sliding mode control \cite{dagci2003path,kumar2024generic,kumar2024robust}, and feedback linearization \cite{voos2009nonlinear,mokhtari2006feedback}, have been developed. The authors in \cite{mohd2015enhanced, lungu2020auto} designed a controller using the backstepping method, where the outer loop (heading rates) guidance command is developed first, followed by the inner loop design (attitude rates). In \cite{dagci2003path,kumar2024generic,wang2022integrated,mofid2024robust}, a sliding surface corresponding to the desired system response was defined, and the system states were driven onto this surface, which leads to the path following behavior. In \cite{jia2017integral,labbadi2019robust}, both backstepping and sliding mode techniques are combined to stabilize the quadrotor attitude and accomplish trajectory tracking under external disturbances. In \cite{voos2009nonlinear,mokhtari2006feedback,yang2022robust}, guidance laws were designed using feedback linearization, where the UAV’s nonlinear dynamics were linearized through an appropriate state transformation and nonlinear feedback. While these approaches improved performance, they typically neglected the actuator saturation effects. 
	
	While research on developing path-following guidance strategies has been extensively explored in the literature, their development under bounded control constraints has received relatively little attention. For instance, a PP-based approach that guaranteed fixed-time convergence was discussed in \cite{doi:10.2514/6.2026-1947,doi:10.2514/1.G008792}. A nested-saturation-based bounded control method is proposed in \cite{teel1992global,johnson2003nested}, and was later applied to design control strategies for quadrotors in \cite{patrikar2019nested}, interceptor guidance in \cite{kathiriya2025impact}, and for point-mass UAVs in \cite{kendoul2006nonlinear}.  A bounded-input path-following guidance law for fixed-wing UAVs under wind disturbances is presented in \cite{beard2014fixed}.
	% Recent studies have demonstrated that nonlinear model predictive control effectively handles actuator constraints during complex maneuvers \cite{ru2017nonlinear}. In addition, \cite{carlos2020efficient} utilizes real-time iteration schemes to maintain bounded inputs even in the presence of communication time delays.
	
	Inspired by the aforementioned works, this work's main contributions are summarized as follows: We propose a bounded-input 3D path-following guidance law for UAVs that ensures convergence to a predefined desired path while adhering to control input constraints. The proposed strategy guarantees an exponential convergence of the UAV to the desired path. Unlike existing methods, which are primarily developed for the straight-line or circular paths, the proposed approach remains applicable to sufficiently smooth paths. A comparative analysis with several existing path-following methods, designed using various control techniques, shows that the proposed strategy requires lower control effort and achieves enhanced tracking performance. Furthermore, the design is simple and can be implemented using inexpensive onboard sensors, making it suitable for real-world applications.   
	
	The rest of the paper is structured as follows. \Cref{sec:problem} describes the UAV kinematics and problem statement. The proposed bounded-input path-following guidance law will be derived using nested saturation theory in \Cref{sec:main}, followed by demonstrating the guidance law's performance through simulation results in \Cref{sec:simulations}. \Cref{sec:conclusion} concludes the work while indicating some future directions of research.  
	
	\section{Problem Formulation}\label{sec:problem}
	Consider a UAV required to follow a three-dimensional predefined path $P' \in \mathbb{R}^3$, which is assumed to be a generic (smooth) curve, as shown in \Cref{fig:3D_general_problem}. The UAV is modeled as a point-mass, nonholonomic vehicle, where the horizontal and vertical accelerations $a_h$ and $a_v$ are used to guide the vehicle along the desired path, while maintaining a constant speed. The kinematics of the UAV are governed by
	\begin{subequations}\label{eq:3D_dynamics}
		\allowdisplaybreaks
		\begin{align}
			\dot{x}&=v\cos\gamma \cos\chi, \label{eq:dot_x}\\
			\dot{y}&=v\cos\gamma \sin\chi,\label{eq:dot_y}\\
			\dot{z}&=v\sin\gamma, \label{eq:dot_z}\\
			\dot{\chi}&=\dfrac{a_h}{v\cos\gamma}, \label{eq:dot_chi}\\
			\dot{\gamma}& =\dfrac{a_v}{v} \label{eq:dot_gamma},
		\end{align} 
	\end{subequations}
	where $[x,\,y,\,z]^\top \in \mathbb{R}^3$ denotes the instantaneous position of the UAV in three-dimensional space, $v \in \mathbb{R}_{+}$ is the constant airspeed with respect to the inertial frame, and $\chi$ and $\gamma$ represent the heading angle and the flight-path angle of the UAV, respectively. Note that the expressions in \eqref{eq:dot_x}-\eqref{eq:dot_z} describe the translational motion of the UAV in the inertial frame of reference, whereas \eqref{eq:dot_chi} and \eqref{eq:dot_gamma} denote the turn rates of the UAV into the mutually orthogonal planes.
	\begin{figure}[!ht]
		\centering
		\includegraphics[width=0.8\linewidth]{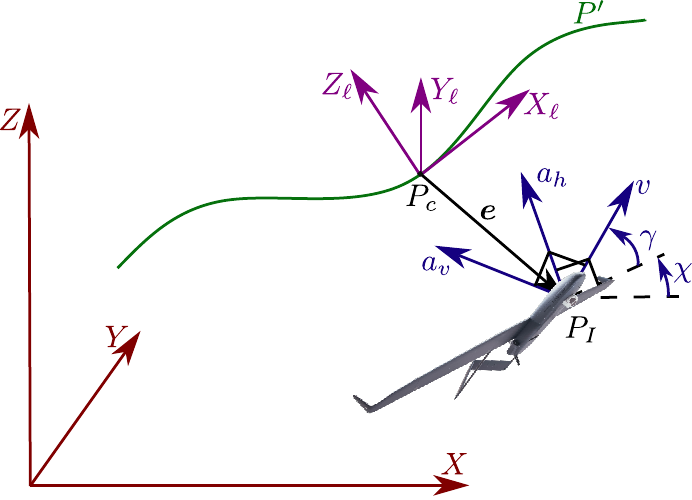}
		\caption{Illustration of 3D path following problem.}
		\label{fig:3D_general_problem}
	\end{figure} 
	
	It is essential to note that most three-dimensional path-following guidance strategies decompose the vehicle motion into two planes in the inertial frame, namely the horizontal ($X$-$Y$) and vertical ($X$-$Z$) planes. While such decompositions may be adequate for simple trajectories, such as straight-line paths, they become ineffective for spatially curved or inclined paths, as they fail to preserve the true geometry of the reference trajectory. To overcome this limitation, we construct a local path-tangent coordinate frame that moves along the desired path. This local moving frame captures the curvature and orientation of the path in three-dimensional space, thereby enabling an accurate computation of cross-track errors and associated control inputs.
	
	Let $P' \in \mathbb{R}^3$ denote the predefined smooth reference path, and let the UAV’s instantaneous position be denoted by $P_I \in \mathbb{R}^3$. The closest point on the path to the UAV is denoted by $P_c$. At this point, a local coordinate frame ${X_\ell Y_\ell Z_\ell}$ is defined, analogous to the Frenet--Serret frame. Here, $X_\ell$ is aligned with the tangent vector to the path, $Y_\ell$ is the normal (radial) vector, and $Z_\ell$ is the binormal vector, defined as $X_\ell \times Y_\ell$. This construction yields a coordinate system that adapts to the instantaneous geometry and orientation of the reference path, providing a convenient basis for defining path-following errors and guidance laws in three-dimensional space.
	
	The cross-track error vector is defined as $\bm{e} = P_I - P_c$. Since $\bm{e}$ is in the inertial frame, it must be decomposed into meaningful components by projecting it onto the local coordinate frame. Projecting $\bm{e}$ onto the each of the axis of the local frame gives the deviation in that direction: the horizontal cross-track error $d_h =\bm{e}^\top \bm{J}_\ell$ measures the lateral displacement of the UAV from the desired path in the horizontal plane, and the vertical cross-track error $d_v =\bm{e}^\top \bm{K}_\ell$ measures the displacement in the vertical plane. Here $\bm{J}_\ell$ and $\bm{K}_\ell$ denote the unit vector along $Y_\ell$ and $Z_\ell$ - axes, respectively. To ensure an accurate path following, the guidance law must generate appropriate control inputs to drive both tracking errors $d_h$ and $d_v$ to zero while satisfying the bounded control constraints $|a_h| \leq a_h^{\max}$ and $|a_v| \leq a_v^{\max}$. For control design purposes, the tracking errors $d_h$ and $d_v$, along with their time derivatives $\dot{d}_h$ and $\dot{d}_v$, are modeled in the local frame. This transformation effectively decouples the guidance problem into two independent subsystems: one defined in the horizontal plane and the other in the vertical plane. As a result, the original three-dimensional path-following problem is reduced to two simpler two-dimensional problems. Convergence of the UAV in each plane guarantees convergence to the desired three-dimensional path. 
	% This geometric framework enables accurate path following while respecting input constraints.
	
	In practical UAV applications, control inputs are inherently bounded by the physical limitations of UAVs and their onboard actuators. However, most existing guidance strategies do not explicitly account for these constraints during the design process. As a consequence, such approaches may suffer from degraded performance or even lead to instability when implemented on actual vehicles. To address this issue, the present work explicitly incorporates the bounded control inputs by imposing constraints on the commanded lateral accelerations. We are now in a position to formally state the main problem addressed in this paper.
	
	\begin{problem}
		For the given engagement geometry between the UAV and the desired path, design the commanded horizontal acceleration $a_h$ and vertical acceleration $a_v$ such that the following objectives are satisfied:
		\begin{itemize}
			\item $d_h \to 0, \dot{d}_h \to 0$ as $t \to \infty$,
			\item $d_v \to 0, \dot{d}_v \to 0$ as $t \to \infty$,
			\item control inputs must satisfy $a_h \in \mathbb{U}_h$ and $a_v \in \mathbb{U}_v$ $\forall$ $t \geq 0$, where $\mathbb{U}_h \coloneqq \{a_h \in \mathbb{R}:|a_h| \leq a_h^{\max}\}$ and $\mathbb{U}_v \coloneqq \{a_v \in \mathbb{R}:|a_v| \leq a_v^{\max}\}$, and $a_h^{\max}$ and $a_v^{\max}$ are strictly positive constants.
		\end{itemize}
	\end{problem}
	
	In other words, we aim to design a nonlinear path-following guidance strategy that ensures the UAV follows its desired path while never violating the control input constraints. Note that the guidance strategy should be derived within a nonlinear engagement framework to ensure its applicability to a wider range of operating conditions. In the following section, we derive the guidance command to achieve the said objectives.
	
	\section{Derivation of Guidance Strategy}\label{sec:main}
	In this section, we first develop a bounded-input path-following guidance law for planar paths and subsequently extend the formulation to the three-dimensional case.
	
	Consider a planar engagement scenario, as shown in \Cref{fig:1}, in which the UAV maintains level flight, that is, the flight-path angle and its rate satisfy $\gamma=\dot{\gamma}=0$. Under this assumption, the UAV kinematics given by \eqref{eq:3D_dynamics} reduces to  
	\begin{align}
		\dot{x} = v \cos\chi,\;
		\dot{y} = v \sin\chi,\;
		\dot{\chi} = \dfrac{a_h}{v}. \label{eq:2D Dynamics}
	\end{align}
	Since the vehicle maintains a level flight, $\dot{z}=0$. For the simplicity of the notation, we denote $\chi$ as $\psi$ and $a_h$ as $a_m$, and thus the UAV's dynamics given in \eqref{eq:2D Dynamics} becomes
	\begin{align}
		\dot{x} = v \cos\psi,\,  
		\dot{y} = v \sin\psi,\     
		\dot{\psi} = \dfrac{a_m}{v}, 
	\end{align}
	where we refer to $\psi$ as the heading angle and $a_m$ as the lateral acceleration that serves as the control input. The lateral acceleration is assumed to be bounded as
	\begin{align}
		|a_m(t)| \le a_{\max}, \, \forall t \ge 0,
	\end{align}
	where $a_{\max} = a_{h}^{\max}>0$ is a known maximum allowable lateral acceleration.  
	\begin{figure}[!ht]
		\centering
		\includegraphics[width=0.7\linewidth]{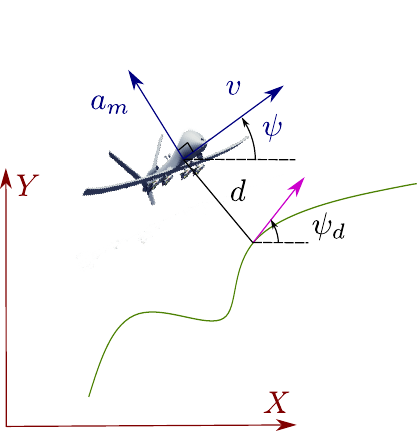}
		\caption{Illustration of planar path following problem.}
		\label{fig:1}
	\end{figure}
	
	Let the term $d$ denote the perpendicular distance between the UAV and a predefined smooth planar path, and $\psi_d$ denote the desired heading angle (see \Cref{fig:1}).  We now derive the UAV's guidance law to ensure it follows the desired planar path, subject to bounded control inputs. To follow any desired planar path, it is essential to first steer the UAV onto it. In our context, this can be achieved by nullifying the perpendicular separation between the UAV and the desired path. Therefore, we consider the perpendicular distance, $d$, as a metric, whose nullification can lead to the desired path-following behavior. To that end, we first carefully analyze the behavior of error $d$. One may obtain the rate of change of $d$ by decomposing the UAV's velocity vector along it as
	\begin{align}
		\dot{d} = v \sin(\psi-\psi_d). \label{eq:d_dot}
	\end{align}
	The dynamics of the distance $d$, given in \cite{kumar2023robust}, can be obtained by differentiating \eqref{eq:d_dot} with respect to time as 
	\begin{align}
		\ddot{d} = a_m \cos(\psi - \psi_d) - v \dot{\psi}_d \cos(\psi -\psi_d). \label{eq:d_ddot}
	\end{align}
	It follows from \eqref{eq:d_ddot} that the dynamics of $d$ has a relative degree of two with respect to the UAV's steering control input, $a_m$.
	\begin{assumption}\label{remark:dpsi_d}
		The path curvature, denoted as $\dot{\psi}_{d}$, is known and satisfies the following inequality
		\begin{align}
			v|\dot{\psi}_d|\leq v\dot{\psi}_d^{\mathrm{max}} \leq a_{\mathrm{max}},
		\end{align}
		where $\dot{\psi}_d^{\mathrm{max}}$ denotes the maximum curvature of desired path.
	\end{assumption}
	Note that this assumption ensures the feasibility of the path that need to be followed by the UAV. Next, define the error states
	\begin{align}
		d_1 \coloneqq d, \, d_2 \coloneqq \dot d.
	\end{align}
	Then, the error dynamics can be written as
	\begin{subequations}\label{eq:d_states}
		\begin{align}
			\dot d_1 &= d_2=v\sin(\psi-\psi_d), \label{eq:dy1} \\
			\dot d_2 &= a_m \cos(\psi - \psi_d) - v \dot{\psi}_d \cos(\psi - \psi_d)\label{eq:dy2}.
		\end{align}
	\end{subequations} 
	By choosing lateral acceleration command in \eqref{eq:d_states} as
	\begin{align}
		a_m = \frac{u + v \dot{\psi}_d \cos(\psi - \psi_d)}{\cos(\psi - \psi_d)},
		\label{eq:am_u}
	\end{align}
	where $u$ is an auxiliary control input, the expression in \eqref{eq:d_states} becomes
	\begin{align}
		\dot d_1 = d_2, \; \dot d_2 = u.
		\label{eq:dynamics}
	\end{align}
	Note that \eqref{eq:dynamics} is in the form of a double integrator system, which can be written in state-space form as
	\begin{align}
		\dot{\bm d} = \bm A_d \bm d + \bm B_d u,
	\end{align}
	where  $\bm{d}$, $\bm{A}_d$, and $\bm{B}_d$ are defined as 
	\begin{align}
		\bm d = \begin{bmatrix} d_1 \\ d_2 \end{bmatrix}, \;
		\bm A_d =
		\begin{bmatrix}
			0 & 1\\
			0 & 0
		\end{bmatrix}, \;
		\bm B_d =
		\begin{bmatrix}
			0\\
			1
		\end{bmatrix}.
	\end{align}
	
	Next, we introduce a coordinate transformation that enables us to place the closed-loop poles at any arbitrary location on the negative real axis. Towards that, let there be two sets $\mathcal{K}_{1} \coloneqq \{ k_1\}$ and  $\mathcal{K}_{2} \coloneqq \{k_1,k_2\}$ such that $\mathcal{K}_{1} \subset \mathcal{K}_{2}$, and define a function $\mathcal{F}^m_{i}(\mathcal{K}_\ell)$ that acts over the sets $\mathcal{K}_\ell$ for $\ell = 1,2$. The function $\mathcal{F}^m_{i}(\cdot)$ is used to generate the product of combinations of elements, taken $m$ at a time from the set $\mathcal{K}_{\ell}$. In essence, the function $\mathcal{F}^m_{i}(\mathcal{K}_\ell)$ is a generating function that gives the output as $i\textsuperscript{th}$ combinations of products of $m$ elements taken from the sets $\mathcal{K}_{\ell}$, without repetition and disregarding order with $\mathcal{F}^0_{i}=1$. We use a transformation to place the closed-loop poles at an arbitrary desired location on the negative real axis. We now introduce the function $\mathcal{C}(l,m)$, defined over the set 
	$\mathcal{K}_{l}$ for integers $l \geq 0$ and $0 \leq m \leq l$. This function is defined as  
	\begin{align}
		\mathcal{C}(l,m) &= \sum_{k=1}^{\mathcal{\bar{C}}^l_m} \mathcal{F}^{m}_{k}(\mathcal{K}_{l}),\; \mathcal{C}(l,0) = 1,
	\end{align}
	where $\mathcal{\bar{C}}^l_m$ is the binomial coefficient of $\left(^l_m\right)$. By using above construction, we define new coordinate system $\bm{h} \coloneqq [h_1, h_2]^\top$ as a transformation from $\bm{d} \coloneqq [d_1,d_2]^\top$, as
	\begin{align}
		h_{2-i}=k_{i+1} \sum_{j=0}^i \mathcal{C}(i,j) d_{2-j},\;\forall\;i=0,1. \label{eq:h}
	\end{align}
	Setting $i=1$ and $0$ in \eqref{eq:h}, we get $h_1$ and $h_2$ as
	\begin{subequations} \label{eq:h1h2}
		\begin{align}
			h_1&=k_1k_2d_1+k_2d_2, \label{eq:h1}\\
			h_2&=k_1d_2. \label{eq:h2}
		\end{align}
	\end{subequations}
	The expression in \eqref{eq:h1h2} can be represented in a compact form, with $\bm{T}_{\rm hd} \in \mathbb{R}^{2 \times 2}$ as a transformation matrix, as 
	\begin{align}
		\bm{h}=\bm{T}_{\rm hd}\bm{d},\bm{T}_{\rm hd}=
		\begin{bmatrix}
			k_1k_2 & k_2\\
			0 & k_1
		\end{bmatrix}.
	\end{align}
	Since $\bm{T}_{\rm hd}$ is an upper triangular matrix with nonzero diagonal entries, it is always invertible, with its inverse given by 
	\begin{equation}
		\bm{T}^{-1}_{\rm hd}= \frac{1}{k_1^2k_2}
		\begin{bmatrix}
			k_1 & -k_2\\
			0 & k_1 k_2
		\end{bmatrix}.\label{eq:t_inverse}
	\end{equation}
	To obtain the dynamics of variable $\bm{h}$, we differentiate  $\bm{h}=\bm{T}_{\mathrm{hd}} \bm{d}$ with respect to time, and use the relations, $\dot{\bm{d}}=\bm{A}_{\mathrm{d}} \bm{d}+\bm{B}_{\mathrm{d}}u$ and $\bm{d} = T^{-1}_{\rm hd} \bm{h} $, to obtain
	\begin{align}
		\dot{\bm{h}}&=\bm{T}_{\rm hd}\dot{\bm{d}}=\bm{T}_{\rm hd}\bm{A}_{\mathrm{d}}d+T_{\rm hd}\bm{B}_{\mathrm{d}} u ,\nonumber  \\
		&=\bm{T}_{\rm hd}\bm{A}_{\mathrm{d}} \bm{T}^{-1}_{\rm hd}\bm{h}+\bm{T}_{\rm hd}\bm{B}_{\mathrm{d}}u=\bm{A}_{\rm h}\bm{h}+\bm{B}_{\rm h}u ,  \label{eq:h_dynamics_1}
	\end{align}
	where $\bm{A}_{\rm h} \in \mathbb{R}^{2 \times 2}$ and $\bm{B}_{\rm h} \in \mathbb{R}^{2 \times 1}$ are defined as 
	\begin{align*}
		\bm{A}_{\mathrm{h}}&=
		\begin{bmatrix}
			k_1k_2 & k_2\\
			0 & k_1  
		\end{bmatrix}
		\begin{bmatrix}
			0&1\\
			0&0
		\end{bmatrix} \frac{1}{k_1^2 k_2}
		\begin{bmatrix}
			k_1&-k_2\\
			0&k_1k_2
		\end{bmatrix}
		=
		\begin{bmatrix}
			0&k_2\\
			0&0
		\end{bmatrix},\\
		\bm{B}_{\mathrm{h}}&=
		\begin{bmatrix}
			k_1k_2 & k_2\\
			0 & k_1  
		\end{bmatrix}
		\begin{bmatrix}
			0\\
			1
		\end{bmatrix}
		=
		\begin{bmatrix}
			k_2\\
			k_1
		\end{bmatrix}.
	\end{align*}
	By expanding \eqref{eq:h_dynamics_1}, we get the dynamics as
	\begin{align}
		\dot{h}_{1}=k_2h_2+k_2u,\;
		\dot{h}_2=k_1u .\label{eq:new_dynamics}
	\end{align}
	
	% \textcolor{blue}{The above transformation was first introduced in \cite{johnson2003nested} to place the closed-loop poles of \eqref{eq:new_dynamics} on the negative real axis. To bound the control input, we employ a nested saturation structure discussed in \cite{teel1992global}.}
	Before deriving the control input $u$, we introduce a definition that facilitates the control design process.
	\begin{definition} (\cite{teel1992global})
		Let there be a positive constants $M \in \mathbb{R}_+$ and a function $\sigma(\cdot):\mathbb{R}\rightarrow \mathbb{R}$. The function $\sigma_i(\cdot)$ is said to be linear saturation if it is continuous, non-decreasing, and satisfies the following conditions
		\begin{subequations}
			\begin{align}
				s\sigma_i(s)&>0; \,\forall s\neq 0, \label{con:1}\\  
				\sigma_i(s) &= \begin{cases}
					s & |s|\leq M_i,\\
					M_i\,\sign(s)&\text{otherwise}.
				\end{cases}\label{con:2}
			\end{align} 
		\end{subequations}
	\end{definition}
	Using the dynamics \eqref{eq:dynamics} and a linear saturation function on $\sigma_i(\cdot)$ defined as above, we now proceed to design the guidance law.
	\begin{theorem}\label{thm:u}
		Consider the system dynamics given in \eqref{eq:new_dynamics}.  Let the gains $M_1$ and $M_2$ are
		\begin{align}
			M_1 < \dfrac{M_2}{2},~~M_2 \coloneqq 
			\left|\left(
			M_2' 
			-  v |\dot{\psi}_d| \right)\cos(\psi-\psi_d)
			\right|,
			\label{eq:M2_def}
		\end{align}
		where $M_2'>a_{\mathrm{max}}$ is a known bound on the lateral acceleration. For any linear saturation functions $\sigma_1(.)$ and $\sigma_2(.)$
		with bounds $M_1$ and $M_2$, respectively, the constrained-input guidance law
		\begin{align}
			u = -\sigma_2\!\left(h_2+\sigma_1(h_1)\right),
			\label{eq:input}
		\end{align}
		ensures an exponential convergence of the error, $d$, and its derivative, $\dot d$, to zero. Consequently, the UAV follows the desired path.
	\end{theorem} 
	
	\begin{proof}
		By using the coordinate transformation $\bm{h} = \bm{T}_{\mathrm{hd}} \bm{d}$, the expression in \eqref{eq:dynamics} can be transformed into \eqref{eq:new_dynamics}. With the guidance law given in \eqref{eq:input}, the closed-loop dynamics become
		\begin{subequations}
			\begin{align}
				\dot{h}_1&=k_2k_1d_2-k_2\sigma_2(h_2+\sigma_1(h_1)), \label{eq:y1_d}\\
				\dot{h}_2&=-k_1\sigma_2(h_2+\sigma_1(h_1)).  \label{eq:y2_d}
			\end{align} \label{eq:y1y2_d}
		\end{subequations}
		We first analyze the dynamics $h_{2}$. Consider the Lyapunov function candidate $\mathcal{V}_2 =\frac{h_2^2}{2}$, whose time derivative along the state trajectories can be obtained as  
		\begin{align}
			\dot{\mathcal{V}}_2 = h_2\dot{h}_2 = \textcolor{blue}{-} h_2k_1\sigma_2(h_2+\sigma_1(h_1)) .
		\end{align}
		Noticing that $k_1$, $k_2 >0$, and using \eqref{con:1} and \eqref{con:2} implies that $h_2$ has the same sign as $\sigma_2(.)$ only if $h_2$ has the same sign as $h_2+\sigma_1(h_{1})$. Upon applying equation \eqref{con:2} to $\sigma_1(h_{1})$ and selecting $M_1 < \frac{M_2}{2}$, it immediately follows that $\dot{\mathcal{V}}_{2}<0\,\forall\,h_2 \notin  \mathcal{D}_2\coloneqq \{h_2:|h_2| \leq M_2/2\}$. Thus, if the trajectory of $h_2$ starts outside of the region $\mathcal{D}_2$, it enters the region $\mathcal{D}_2$ within a finite time and remains in $\mathcal{D}_2$ for all future times. Note that the right-hand side of \eqref{eq:y1y2_d} is globally Lipschitz, and that the derivatives are bounded, implying that the state $h_1$ remains bounded for any finite time.
		
		Once $h_2$ enters the region $\mathcal{D}_2$, \eqref{con:2} implies $\sigma_2(.)$ operates in its linear region since the arguments of $\sigma_2$ is upper bounded as
		\begin{align}
			&|h_2+\sigma_1(h_1)| \leq \frac{M_2}{2} +M_1 <M_2. \label{eq:s21}
		\end{align}
		From the definition of $M_2$ and \Cref{remark:dpsi_d}, we have
		\begin{align}
			M_2 =\lvert (M_2'- v|\dot{\psi}_d|)\cos(\psi-\psi_d) \rvert \leq M_2'. \label{eq:s22}
		\end{align}
		Thus, from \eqref{eq:s21} and \eqref{eq:s22}, it implies that $\sigma_2(.)$ is bounded by $M_2'$. Now, the dynamics of $h_1$ can be obtained as
		\begin{align}
			\dot{h}_1&=k_2k_1d_2-k_2(h_2+\sigma_1(h_1)), \nonumber\\
			&= k_2k_1d_2-k_2k_1d_2-k_2\sigma_1(h_1))=-k_2\sigma_1(h_1) .
		\end{align}
		Next, we define another Lyapunov function candidate as $\mathcal{V}_1=({h_1^2}/{2})$, whose time derivative can be obtained as
		\begin{align}
			\dot{\mathcal{V}}_1=h_1\dot{h}_1=-k_2h_1\sigma_1(h_1) .
		\end{align}
		Since $k_2>0$, using \eqref{con:1}, we have $\dot{\mathcal{V}}_1<0$ for all $h_1\notin \mathcal{D}_1=\{h_1:|h_1|\leq M_1\}$. As a result, $h_1$ enters into region $\mathcal{D}_1$ within a finite time, and subsequently it remains confined to $\mathcal{D}_1$ for all future time. Thus, both the states $h_2$ and $h_1$ enter into the sets $\mathcal{D}_2$ and $\mathcal{D}_1$ within a finite time, and the saturation functions $\sigma_2(.)$ and $\sigma_1(.)$ are operating in a linear region. This implies that after a certain finite time, the closed-loop dynamics become
		\begin{subequations}
			\begin{align}
				\dot{h}_1&=-k_2h_1  ,\\
				\dot{h}_2&=-k_1(h_1+h_2).
			\end{align}    \label{eq:nnnew}
		\end{subequations}
		To show the exponential stability, we express the closed-loop dynamics \eqref{eq:nnnew} in a compact matrix form as
		\begin{align}
			\dot{\bm{h}} = \bm{Q} \bm{h},
		\end{align}
		where the matrix $\bm{Q} \in \mathbb{R}^{2 \times 2}$ is given by
		\begin{align}
			\bm{Q}=
			\begin{bmatrix}
				-k_2 & 0 \\
				-k_1 & -k_1
			\end{bmatrix}.
		\end{align}
		The eigenvalues of the matrix $\bm{Q}$ are computed by solving the characteristic equation $ \det(\Lambda  \bm{I}_2 - \bm{Q}) =0$, where $\bm{I}_2\in\mathbb{R}^{2\times 2}$ represents an identity matrix. The matrix $(\Lambda \bm{I}_2 - \bm{Q})$ can be evaluated as  
		\begin{align*}
			\Lambda \bm{I}_{2}-\bm{Q}=
			\begin{bmatrix}
				\Lambda & 0\\
				0 & \Lambda  
			\end{bmatrix}
			-
			\begin{bmatrix}
				-k_2 & 0\\
				-k_1 & -k_1  
			\end{bmatrix} 
			=
			\begin{bmatrix}
				\Lambda+k_2 & 0\\
				k_1 & \Lambda+k_1  
			\end{bmatrix},
		\end{align*}
		Setting the determinant $\det(\Lambda  \bm{I}_2 - \bm{Q}) =0$ yields
		\begin{align}
			\mathrm{det}|\Lambda \bm{I}_{2}-\bm{Q}|&=(\Lambda+k_1)(\Lambda+k_2)=0. \label{eq:eigen_vales}
		\end{align}
		On solving \eqref{eq:eigen_vales}, we get $\Lambda=-k_1,-k_2$. Since the constants $k_1$, $k_2>0$, all the eigenvalues of the matrix $\bm{Q}$ are strictly negative real. Once both $\sigma_1(.)$ and $\sigma_2(.)$ simultaneously operate in the linear region, $h_1$ and $h_2$ are exponentially stable and converge to zero. 
		
		To demonstrate that $d_1$ and $d_2$ are also exponentially stable, we use the transformation $\bm{h} = \bm{T}_{\mathrm{hd}}\bm{d}$. Since $\bm{T}_{\mathrm{hd}}$ is a square matrix whose determinant is given by $\det(\bm{T}_{\mathrm{hd}}) = k_1^2k_2 \neq 0$, implying its inverse will always exist, given by \eqref{eq:t_inverse}. Therefore, $\bm{d} = \bm{T}_{\mathrm{hd}}^{-1}\bm{h}$, is also exponentially stable. As a matter of fact, $d_1$ and $d_2$ exponentially converge to zero. This completes the proof.
	\end{proof} 
	
	Next, we calculate the closed-loop pole location when $\sigma_2(.)$ is unsaturated and $\sigma_1(.)$ is saturated, or both are unsaturated. We start with the former case when $\sigma_1(.)$ is saturated and $\sigma_2(.)$ is unsaturated.
	By substituting the guidance command from \eqref{eq:input} into \eqref{eq:dynamics}, one may obtain
	\begin{align}
		\dot{d}_2=u=-\sigma_2(h_2+\sigma(h_1)). \label{eq:dotd2_1}
	\end{align}
	For the case when $\sigma_1(.)$ is saturated and $\sigma_2(.)$ is unsaturated, \eqref{eq:dotd2_1} becomes
	\begin{align*}
		\dot{d}_2=-h_2\pm M_1 \implies  \dot{d}_2+h_2\mp M_1=0,
	\end{align*}
	which represents a forced linear system with a constant forcing function $M_1$. The corresponding homogeneous part is given by  
	\begin{align}
		\dot{d}_{2}+h_2=0 . \label{eq:dotd2_2}
	\end{align}
	By substituting equation \eqref{eq:h2} and \eqref{eq:dynamics} into \eqref{eq:dotd2_2}, one may obtain
	\begin{align}
		\ddot{d}_{1}+k_1\dot{d}_{1}=0. \label{eq:char}
	\end{align}
	The characteristic equation for \eqref{eq:char} is 
	\begin{align}
		\gamma(\Lambda)=\Lambda^2+k_1\Lambda=\Lambda(\Lambda+k_1),
	\end{align}
	By equating the characteristic equation with zero, one may obtain closed-loop poles at \{$-k_1,0$\}.
	
	Next, we calculate the location of closed-loop poles for the case when both $\sigma_1(.)$ and $\sigma_2(.)$ are unsaturated. For such a scenario, \eqref{eq:dotd2_1} becomes
	\begin{align}
		\dot{d}_2=-h_2-h_1.  \label{eq:h1_h2}
	\end{align}
	By substituting \eqref{eq:h1} and \eqref{eq:h2} into \eqref{eq:h1_h2}, one may obtain
	\begin{align}
		\dot{d}_2+k_1d_2+k_1k_2d_1+k_1d_2&=0 , \nonumber\\    
		\dot{d}_2+(k_1+k_2)d_2+k_1k_2d_1&=0, \label{dd}
	\end{align} 
	which on further substituting \eqref{eq:dy1} and \eqref{eq:dy2} into \eqref{dd}, lead us to arrive at
	\begin{align}
		\ddot{d}_1+(k_1+k_2)\dot{d}_1+k_1k_2d_1=0, \label{eq:char2}
	\end{align}
	By following a similar procedure to the previous case, one can obtain the location of closed-loop poles at \{$-k_1,-k_2$\}. 
	
	Note that \Cref{thm:u} establishes that the transformed control input $u$ remains bounded by $M_{2}$. Next, we endeavor to prove that $a_m$ remains bounded in the following corollary.
	\begin{corollary}
		If $u$ remains bounded by $M_{2}$, then the UAV's steering control $a_m$ remains bounded by $\pm M_2'$ for all time $t>0$.
	\end{corollary}
	\begin{proof}
		Taking absolute value on both sides of \eqref{eq:am_u} yields
		\begin{align}
			|a_m|=\left|  \frac{u + v \dot{\psi}_d \cos(\psi - \psi_d)}{\cos(\psi - \psi_d)}  \right|. \label{eq:mode_a_m}
		\end{align}
		Using the triangle equality in the numerator of \eqref{eq:mode_a_m}, we obtain 
		\begin{align}
			|a_m| \leq  \frac{\left|u\right| + \left|v \dot{\psi}_d \cos(\psi - \psi_d)\right|}{\left|\cos(\psi - \psi_d)\right|}. \label{eq:mode_2}
		\end{align}
		Using the fact that the maximum value of $|u| \leq M_2$, the expression in \eqref{eq:mode_2} becomes
		\begin{align}
			|a_m| \leq  \frac{M_2 + |v \dot{\psi}_d \cos(\psi - \psi_d)|}{|\cos(\psi - \psi_d)|}. \label{eq:mode_3} 
		\end{align}
		Recall from \Cref{thm:u} that $ M_2=|(M_2' -  v |\dot{\psi}_d| )\cos(\psi-\psi_d)|$. Substituting the value of $M_2$ into \eqref{eq:mode_3} and separating terms using triangle equality, one may obtain
		\begin{align}
			|a_m| &\leq  \frac{\left|(M_2' -  v |\dot{\psi}_d| )\cos(\psi-\psi_d)
				\right| + \left|v \dot{\psi}_d \cos(\psi - \psi_d)\right|}{|\cos(\psi - \psi_d)|}, \nonumber\\
			&\leq \frac{(|M_2' -  v |\dot{\psi}_d||)|\cos(\psi-\psi_d)| + v (| \dot{\psi}_d |) | \cos(\psi - \psi_d)|}{|\cos(\psi - \psi_d)|}, \nonumber\\
			&\leq |M_2' -  v |\dot{\psi}_d| | + v | \dot{\psi}_d | . 
		\end{align}
		From the definition of $M_2$ and \Cref{remark:dpsi_d}, we have
		\begin{align}
			|a_m| &\leq 
			M_2' 
			-  v |\dot{\psi}_d|  + v | \dot{\psi}_d | =M_2'. \label{eq:am_bound}
		\end{align}
		It follows from \eqref{eq:am_bound} that $|a_m| \leq M_2'$ holds for all $t > 0$, under the proposed control law \eqref{eq:input}.  
	\end{proof}
	
	It can be observed from \eqref{eq:am_u} that lateral acceleration becomes singular when $(\psi - \psi_d) = \pm \pi/2$. Towards that, we now show that $(\psi - \psi_d) = \pi/2$ is not an attractor and therefore it will not cause any singularity during implementations. Let $ \zeta \coloneqq \psi-\psi_d$. To prove that $\zeta = \pm(\pi/2)$ is not an attractor, we have to show that at $\zeta = \pm (\pi/2)$, $\dot{\zeta} \neq 0$. As a matter of fact, even if $\zeta = \pm (\pi/2)$ momentarily, it will not remain there. 
	
	Towards that, we analyze the dynamics of $\zeta$, given by
	\begin{align}
		\dot{\zeta}&=\dot{\psi}-\dot{\psi}_d=\frac{a_m}{v}-\dot{\psi}_d=\frac{u + v \dot{\psi}_d \cos\zeta}{v \cos\zeta}-\dot{\psi_d}, \nonumber\\
		&=\frac{u}{v \cos\zeta}= -\frac{\sigma_2\!\left(k_1d_2+\sigma_1(k_1k_2d_1+k_2d_2)\right)}{v \cos\zeta}. \label{eq:dzeta_0}
	\end{align}
	Note that the behavior of $\dot{\zeta}$ at $\zeta = \pm (\pi/2)$ depends on saturation functions $\sigma_2(.)$ and $\sigma_1(.)$. We analyze three distinct cases that collectively cover all possible operating modes: (i) when $\sigma_2(.)$ is saturated (outer loop saturation) (ii) when $\sigma_2(.)$ is unsaturated and $\sigma_1(.)$ is saturated (inner loop saturation) (iii) when both $\sigma_2(.)$ and $\sigma_1(.)$ are unsaturated (linear operation). 
	
	\textit{Case 1:} When the outer saturation function $\sigma_2(.)$ is saturated.
	Once $\sigma_2(.)$ is saturated, we have $\sigma_2(.)=\pm M_2$, where the sign depends on the direction of saturation. Substituting this into \eqref{eq:dzeta_0} yields 
	\begin{align}
		\dot{\zeta}= \frac{\mp M_2}{v \cos\zeta}=\frac{\mp \left|\left(
			M_2' 
			-  v |\dot{\psi}_d| \right)\cos\zeta
			\right|}{v \cos\zeta} = \frac{\mp \left(|
			M_2' 
			-  v |\dot{\psi}_d| |\right) |\cos\zeta
			|}{v \cos\zeta}. \label{eq:dzeta_1}
	\end{align}
	We analyze the behavior when $\zeta \rightarrow \dfrac{\pi}{2}$ as
	\begin{align}
		\lim_{\zeta \to \pi/2} \dot{\zeta} = \lim_{\zeta \to \pi/2} \frac{\mp \left(|
			M_2' 
			-  v |\dot{\psi}_d| | \right) |\cos\zeta
			|}{v \cos\zeta}. \label{eq:dzeta_2}
	\end{align}
	The evaluation of this limit depends on the behavior of the sign of $\cos\zeta$ in the neighborhood of $\pi/2$. We take two cases when $\zeta$ approaches from left $\left(\zeta \rightarrow\dfrac{\pi^-}{2}\right)$ and from right $\left(\zeta\rightarrow\dfrac{\pi^+}{2}\right)$.
	
	\textit{Subcase 1a:} When $\zeta \to \dfrac{\pi}{2}^{-}$. In this region, $|\cos\zeta| = \cos\zeta$, and thus \eqref{eq:dzeta_2} becomes
	\begin{align}
		\lim_{\zeta \to \pi/2^-} \dot{\zeta} = \lim_{\zeta \to \pi/2^-} \frac{\mp \left|
			M_2' 
			-  v |\dot{\psi}_d| \right|\cos\zeta
		}{v \cos\zeta}. \label{eq:dzeta_3}
	\end{align}
	At $\zeta= \pi/2$, both numerator and denominator approach zero, yielding an indeterminate form $\dfrac{0}{0}$. Applying L'Hôpital's rule by differentiating the numerator and denominator of \eqref{eq:dzeta_3}with respect to $\zeta$ yields 
	\begin{align}
		\lim_{\zeta \to \pi/2^-} \dot{\zeta} &= \lim_{\zeta \to \pi/2^-} \frac{\mp\left(|M_2' -  v|\dot{\psi}_d|| \right) (-\sin\zeta)}{v(-\sin\zeta)} \nonumber \\
		&= \frac{\mp|M_2' -  v|\dot{\psi}_d||}{v}>0, \label{eq:dzeta_3_result}
	\end{align}
	At $\zeta=\pi/2$, $\dot{\zeta}$ has a finite and non-zero value.
	
	\textit{Subcase 1b:} When $\zeta \to \dfrac{\pi}{2}^+$. In this region, $|\cos\zeta| = -\cos\zeta$, so \eqref{eq:dzeta_2} becomes
	\begin{align}
		\lim_{\zeta \to \pi/2^+} \dot{\zeta} = \lim_{\zeta \to \pi/2^+} \frac{\pm|M_2' -  v|\dot{\psi}_d|| \cdot \cos\zeta}{v\cos\zeta}. \label{eq:dzeta_4}
	\end{align}
	Again applying \text{L'H\^opital's} rule
	\begin{align}
		\lim_{\zeta \to \pi/2^+} \dot{\zeta} &= \lim_{\zeta \to \pi/2^+} \frac{\mp\left(|M_2' -  v|\dot{\psi}_d|| \right)  (-\sin\zeta)}{v(-\sin\zeta)} \nonumber \\
		&= \frac{\pm|M_2' -  v|\dot{\psi}_d||}{v}>0, \label{eq:dzeta_4_result}
	\end{align}
	which is also finite and non-zero.
	
	\textit{Case 2:} When $\sigma_2(.)$ is unsaturated and $\sigma_1(.)$ is saturated.  In this case, the outer saturation function operates in the linear region while the inner saturation function is saturated $\sigma_1(.)=\pm M_1$. Then,  \eqref{eq:dzeta_0} reduces to  
	\begin{align}
		\dot{\zeta}&= -\frac{k_1d_2 \pm M_1}{v \cos\zeta}. \label{eq:dzeta_5}
	\end{align}
	From \eqref{thm:u}, we can write $M_1 =M_2/\eta$, where $\eta>2$. Substituting the expression for $d_2=v\sin\zeta$ from \eqref{eq:dy1} and $M_1 = |(M_2' -  v|\dot{\psi}_d|)\cos\zeta|/\eta$ into \eqref{eq:dzeta_5}, we obtain
	\begin{align}
		\dot{\zeta}&= -\frac{k_1v\sin\zeta \pm \dfrac{\left|\left(
				M_2' 
				-  v |\dot{\psi}_d| \right)\cos\zeta
				\right|}{\eta}}{v \cos\zeta} .\label{eq:dzeta_6}
	\end{align}
	Tacking limit as $\zeta \rightarrow \pi/2$ yields  
	\begin{align}
		\lim_{\zeta \to \pi/2} \dot{\zeta} = \lim_{\zeta \to \pi/2} \left[-\frac{k_1 v\sin\zeta \pm \dfrac{|(M_2' -  v|\dot{\psi}_d|)\cos\zeta|}{\eta}}{v\cos\zeta}\right]. \label{eq:dzeta_7}
	\end{align}
	
	Following the same procedure as Case 1, we consider two cases, $\cos\zeta>0$ and $\cos\zeta<0$. In both scenarios, applying \text{L'H\^opital's} rule to solve the $\dfrac{0}{0}$ indeterminate forms yields a finite, non-zero value for $\dot{\zeta}$.
	
	\textit{Case 3:} When both $\sigma_1(.)$ and $\sigma_2(.)$ are unsaturated. When both saturation function operates at the linear region, the \eqref{eq:dzeta_0} becomes 
	\begin{align}
		\dot{\zeta} &= -\frac{k_1d_2+k_1k_2d_1+k_2d_2}{v \cos\zeta} =-\frac{(k_1+k_2)v\sin\zeta+k_1k_2d_1}{v \cos\zeta}. \label{eq:dzeta_8}
	\end{align}
	It can be observed from \eqref{eq:dzeta_8} that when $\zeta \to \pm (\pi/2)$, $\dot{\zeta} \to \infty$. Consequently, $\zeta = \pm (\pi/2)$ is not an attractor for this case as well. Similarly, we can prove the same for $\zeta=-\pi/2$ also.
	
	\begin{figure*}[!ht]
		\centering
		\begin{subfigure}{0.33\linewidth}
			\centering
			\includegraphics[width=0.70\textwidth]{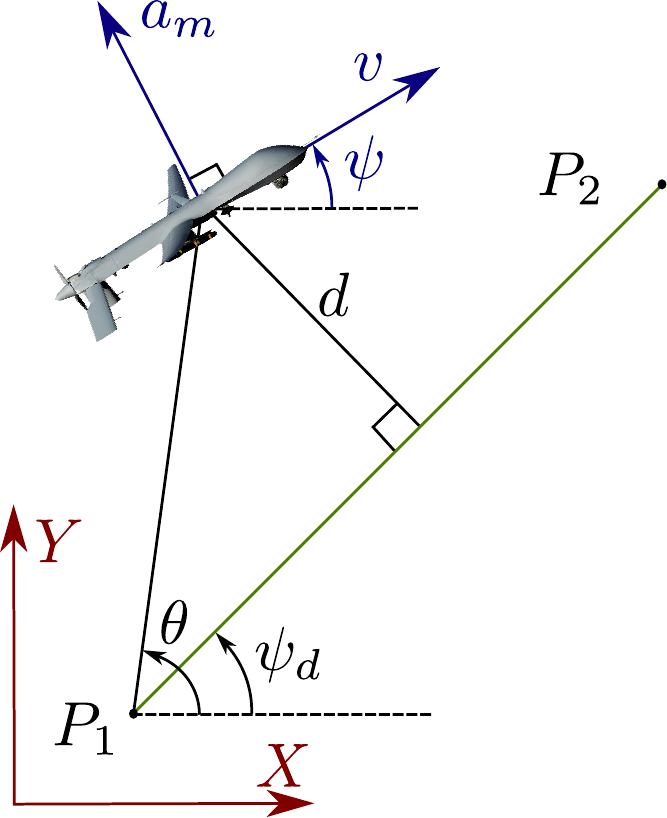}
			\caption{Straight-line path}
			\label{fig:2a}
		\end{subfigure}%
		\begin{subfigure}{0.33\linewidth}
			\centering
			\includegraphics[width=0.65\textwidth]{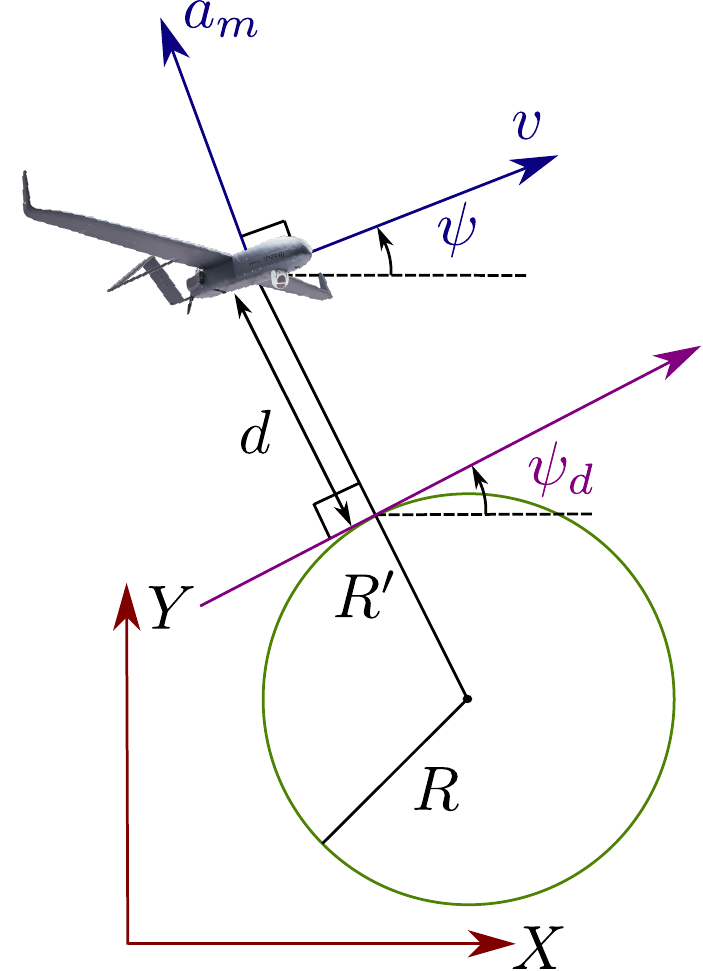} 
			\caption{Circular orbit}
			\label{fig:2b}
		\end{subfigure}
		\begin{subfigure}{0.33\linewidth}
			\centering
			\includegraphics[width=\textwidth]{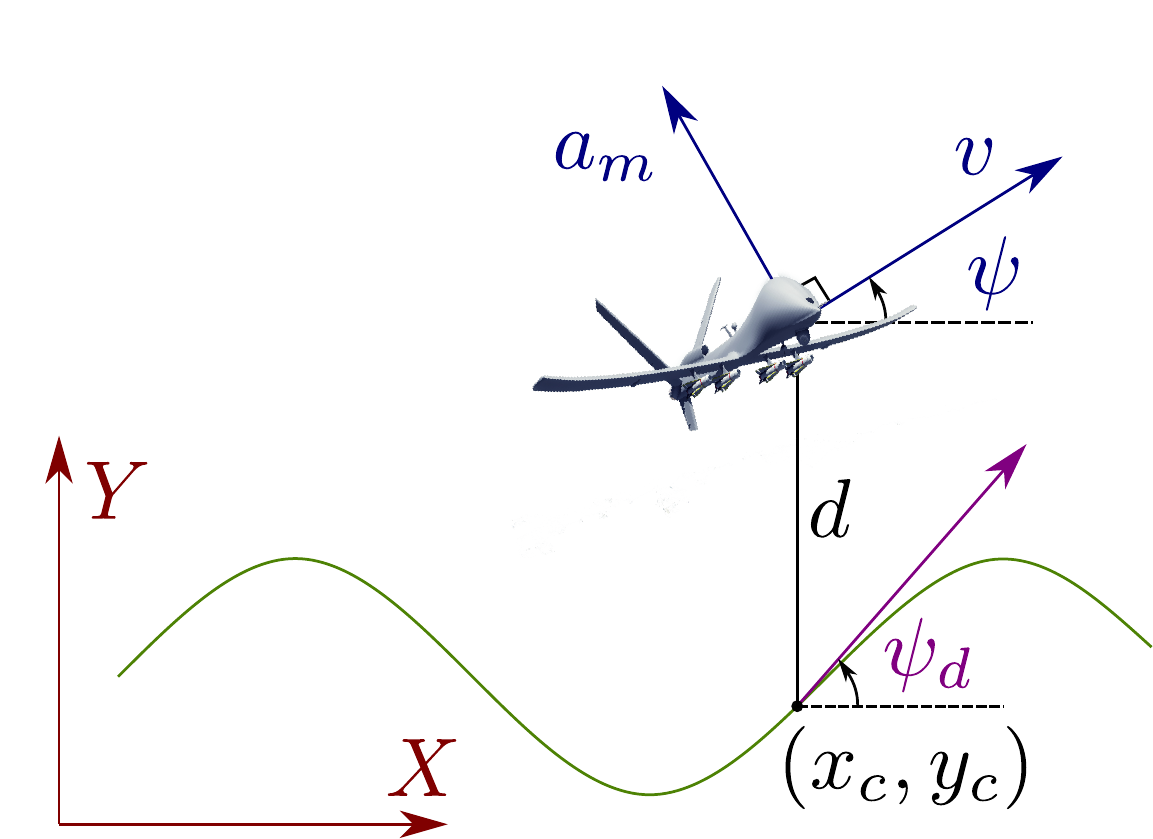} 
			\caption{Sinusoidal path}
			\label{fig:2c}
		\end{subfigure}
		\caption{Planar generic paths.}
		\label{fig:paths}
	\end{figure*}
	
	We now present three special path constructs, namely straight-line, orbital, and sinusoidal paths, as shown in \Cref{fig:paths}. The straight-line path formed by joining two points, $P_1$ and $P_2$, in Euclidean space, as shown in \Cref{fig:2a}. Without loss of generality, we consider the UAV is located a distance of $R$ from point $P_1$ and subtends an angle $\theta$ from the reference. Under such a scenario, one can calculate the distance $d$ using the geometry as $d=R\sin(\theta-\psi_d)$, 
	and its rate of change is given by
	\begin{equation}\label{eq:d_dot_1}
		\dot{d}=v\sin(\psi-\psi_d). 
	\end{equation}
	
	For the orbital path, we consider a circle with center $C\in \mathbb{R}^2$ and radius $R\in \mathbb{R}_+$. The UAV is located at a distance $R'$ from the circle's center. The perpendicular distance $d$ between the UAV and the circular path and its rate are given by
	\begin{align}
		d=\left(R'-R\right), \, \dot{d}=v\sin(\psi-\psi_d), \label{eq:d_dot_test}
	\end{align}
	where $\psi_d$ is the angle made by the tangent vector at $P_c$ with the reference $x$-axis. The point $P_c$ is a point on a circular orbit, which is obtained by where the line $R'$ cuts the orbit. By geometry, one can obtain the angle $\psi_d$ as  
	\begin{align}
		\psi_d=\tan^{-1}\left( \frac{-x}{y} \right),
	\end{align}
	where $(x,y)$ is UAV's instantaneous position.
	
	We now consider a generic sinusoidal path as shown in \Cref{fig:2c}, with amplitude $A$ and frequency $\omega_c$. Note that the sinusoidal path can serve as a basis for any smooth path. The parametric equation of the curve is given by
	\begin{align}
		y_{c} = A \sin(\omega x_{c}), \label{eq:sinecurve}
	\end{align}
	where $x_{c}$ and $y_{c}$ denote the $x$ and $y$ coordinates of the sinusoidal curve. To compute the perpendicular distance $d$ from the UAV to the path, we use a local linear approximation of the curve at the projected $x$-position. The curve is treated as a straight line near the projected $x_{c}$ coordinate, and $d$ is calculated using the formula for the perpendicular distance between a point and a line
	\begin{align}
		d = \frac{y - y_{c}}{\sqrt{1 + \left(\frac{dy_{c}}{dx_{c}}\right)^2}},
	\end{align}
	where the slope of the curve is
	\begin{align}
		\frac{dy_{c}}{dx_{c}} = A \omega \cos(\omega x_c).
	\end{align}
	The rate of change of $d$ will remains the same as in \eqref{eq:d_dot_test} and $\psi_d$ is obtained as
	\begin{align}
		\psi_d = \tan^{-1}\left (\frac{dy_{c}}{dx_{c}} \right).
	\end{align}
	It is important to note that while $d$ depends on the geometry of the path, $\dot{d}$ remains the same for all curves. This property indicates that the proposed technique is applicable to any generic smooth path. 
	
	% \subsection{Three-dimensional path following guidance}
	We now derive guidance law for any generic smooth 3D path. As we decoupled the dynamics into two 2D planes (horizontal and vertical), we derived the guidance law by following the same procedure as for the 2D path-following case. First, we focused on deriving the guidance law for the horizontal plane. To design a controller, we define horizontal plane states as horizontal cross-track error, $d_{1,h}\coloneqq d_h$, and its time derivative, $d_{2,h}\coloneqq \dot{d}_h$. The rate of change of horizontal cross-track error $\dot{d}_h$ is determined by the component of the UAV's velocity vector orthogonal to the path tangent in the local horizontal plane, and the geometric relation is described as     
	\begin{align}
		\dot{d}_{1,h}=d_{2,h}= v\cos\gamma\sin(\chi-\chi_d), \label{eq:dot d_1h}
	\end{align}
	where $\chi_d$ is the desired heading angle. By differentiating \eqref{eq:dot d_1h} with respect to time, we obtained $\dot{d}_{2,h}$ as
	\begin{align}
		\dot{d}_{2,h} &= v \cos\gamma \cos(\chi - \chi_d) (\dot{\chi} - \dot{\chi}_d) , \\
		&= v \cos\gamma \cos(\chi - \chi_d) \left( \frac{a_h}{v \cos\gamma} - \dot{\chi}_d \right)  ,\\
		&= \cos(\chi - \chi_d) a_h - v \cos\gamma\,\dot{\chi}_d \cos(\chi - \chi_d). \label{eq:dot d_2h}
	\end{align}
	In order to make error dynamics into double integrator form, we choose $a_h$ as
	\begin{align}
		a_h = \frac{u_h + v \cos\gamma \dot{\chi}_d \cos(\chi - \chi_d)}{\cos(\chi - \chi_d)},
	\end{align}
	where $u_h$ is the new horizontal plane control input. We are now in a position to define our constrained input guidance law for the horizontal plane. The following theorem proposes the guidance law that guarantees the horizontal error converges to zero.
	\begin{theorem}
		Consider the horizontal plane cross-track error dynamics given in  \eqref{eq:dot d_1h} and \eqref{eq:dot d_2h}. Let $M_{2,h}$ be defined as
		\begin{align}
			M_{2,h}= \left|\left(M_{2,h}'-v|\cos(\gamma)\dot{\chi}_d|\right)\cos(\chi-\chi_d)\right| ,
		\end{align} 
		where $M_{2,h}'>a_h^{\max}$  is a prescribed upper bound on horizontal acceleration $a_h$, and $M_{1,h}$ is chosen such that $M_{1,h} \leq {M_{2,h}}/{2}$. If the bounded guidance law
		\begin{align}
			u_h = -\sigma_{2,h}\left( k_{1,h}\dot{d}_h 
			+ \sigma_{1,h}\big( k_{1,h}k_{2,h}d_h + k_{2,h}\dot{d}_h \big) \right),
		\end{align}
		is applied to horizontal plane, where $k_{1,h}, k_{2,h} > 0$ are positive constant control gains, and 
		$\sigma_{1,h}(\cdot)$ and $\sigma_{2,h}(\cdot)$ are linear saturation 
		functions with saturation levels $M_{1,h}$ and $M_{2,h}$, respectively, 
		then the control input satisfies $|a_h| \leq M_{2,h}'$ for all $t \ge 0$, and the horizontal cross-track error 
		exponentially converges to zero.
	\end{theorem}
	\begin{proof}
		The proof of the above theorem is similar to that of \Cref{thm:u}, thus omitted here. 
	\end{proof}
	Similarly, we derive guidance law for the vertical plane. First, we define our vertical plane states as vertical cross-track error, $d_{1,v}\coloneqq d_v$, and its time derivative, $d_{2,v}\coloneqq \dot{d}_v$. The time derivative of vertical cross-track error is described as
	\begin{align}
		\dot{d}_v = v \sin(\gamma - \gamma_d), \label{eq:dot d_1v}
	\end{align}
	where $\gamma_d$ is the desired flight path angle. By differentiating \eqref{eq:dot d_1v} with respect time, we obtained $\dot{d}_{2,v}$ as
	\begin{align}
		\dot{d}_{2,v} &= v \cos(\gamma - \gamma_d) (\dot{\gamma} - \dot{\gamma}_d),\\
		&= \cos(\gamma - \gamma_d) a_v - v \cos(\gamma -\gamma_d)\dot{\gamma}_d. \label{eq:dot d_2v}
	\end{align}
	We now make the above dynamics into double integrator form by choosing $a_v$ as
	\begin{align}
		a_v = \frac{u_v + v \dot{\gamma}_d \cos(\gamma-\gamma_d)}{\cos(\gamma-\gamma_d)},
	\end{align}
	where $u_v$ is the new vertical plane control input. In the following theorem, we propose the guidance law that guarantees the vertical plane cross-track error converges to zero.
	\begin{theorem}
		Consider vertical plane cross-track error dynamics given by \eqref{eq:dot d_1v}, \eqref{eq:dot d_2v}. Let $M_{2,v}$ be defined as
		\begin{align}
			M_{2,v}= \left|\left(M_{2,v}'\cos(\gamma-\gamma_d)-v|\dot{\gamma}_d|\cos(\gamma-\gamma_d)\right)\right|,
		\end{align}
		where $M_{2,v}'> a_v^{\max} $ is the maximum bound on vertical acceleration $a_v$, and $M_{1,v}$ is chosen as $M_{1,v}\leq {M_{2,v}}/{2}$. If the bounded guidance law
		\begin{align}
			u_v = -\sigma_{2,v}\Big( k_{1,v}\dot{d}_v 
			+ \sigma_{1,v}\big( k_{1,v}k_{2,v}d_v + k_{2,v}\dot{d}_v \big) \Big),
		\end{align}
		is applied to the vertical plane, where $k_{1,v}, k_{2,v} > 0$ are positive constant control gains, and 
		$\sigma_{1,v}(\cdot)$ and $\sigma_{2,v}(\cdot)$ are linear saturation 
		functions with saturation levels $M_{1,v}$ and $M_{2,v}$, respectively, 
		then the control input satisfies $|a_v| \leq M_{2,v}'$ for all $t \ge 0$, and the vertical cross-track errors exponentially 
		converge to zero.
	\end{theorem}
	\begin{proof}
		The proof follows a similar arguments to that of  \Cref{thm:u}, therefore, omitted here.
	\end{proof}
	
	\begin{figure*}[!ht]
		\centering
		\begin{subfigure}{0.48\linewidth}
			\centering
			\includegraphics[width=0.65\textwidth]{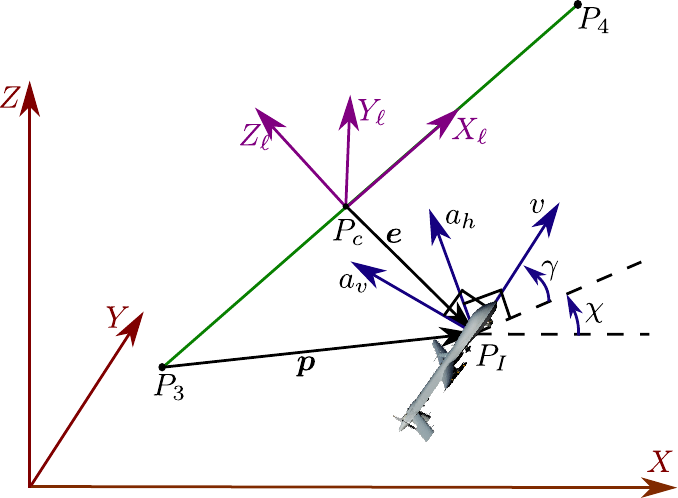}
			\caption{Straight-line path}
			\label{fig:3D_straightline_poblem}
		\end{subfigure}%
		\begin{subfigure}{0.4\linewidth}
			\centering
			\includegraphics[width=1\textwidth]{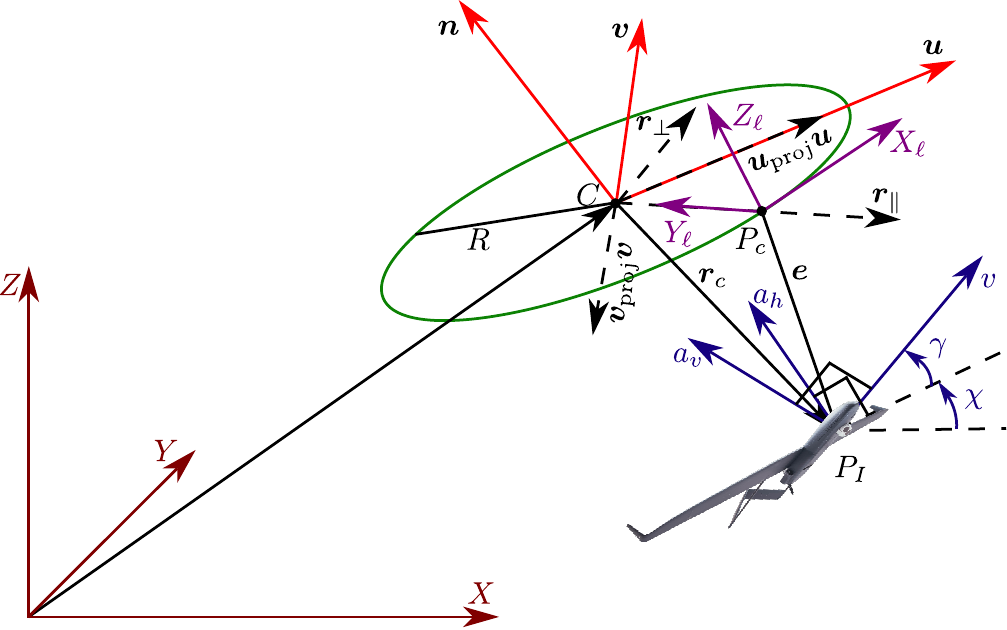} 
			\caption{Circular orbit}
			\label{fig:3D_circular_poblem}
		\end{subfigure}
		\caption{3D generic paths.}
		\label{fig:3D_generic_paths}
	\end{figure*}
	
	We now present two special 3D paths, namely straight-line and orbital paths, as shown in \Cref{fig:3D_generic_paths}. The straight-line path is constructed by joining two points, $P_3 \in \mathbb{R}^3$ and $P_4 \in \mathbb{R}^3$, in the 3-dimensional Euclidean space, as shown in \Cref{fig:3D_straightline_poblem}. The instantaneous position of the UAV is $P_I=\begin{bmatrix}
		x&y&z
	\end{bmatrix}^\top \in \mathbb{R}^3$. To find the cross-track errors, we first determine the point $P_c$ on the straight line closest to the UAV's instantaneous position. The direction of path is defined by unit vector $\bm{u}$, given by
	\begin{align}
		\bm{u}=\frac{P_4-P_3}{||P_4-P_3||}.
	\end{align}
	We define a new vector $\bm{p}$ from the point $P_3$ to the UAV's current position $P_I$ as $\bm{p}=P_I-P_3$. The projection of this vector onto the straight-line path yields $\bm{p}_{\mathrm{proj}}= \bm{u}^\top \bm{p}$. Consequently, the closest point on line $P_c$ is described as
	\begin{align}
		P_c=P_3+\bm{p}_{\mathrm{proj}}\bm{u}.
	\end{align}
	The position error vector, $\bm{e}$, is defined as the vector difference between the UAV's current position and the closest point $P_c$, described as $\bm{e}=P_I-P_c=
	\begin{bmatrix}
		e_1& e_2& e_3
	\end{bmatrix}^\top.$ To obtain horizontal and vertical errors in local frame, a path-fixed local coordinate frame $X_\ell Y_\ell Z_\ell$ is constructed at the closest point $P_c$ on the straight-line path. The first basis vector $X_\ell$ is defined as the unit tangent vector to the path at point $P_c$, that is $X_\ell = \bm{u}$. The second basis vector $Y_\ell$ is constructed to lie in the horizontal plane and is perpendicular to the path tangent direction $X_\ell$. To achieve this, the horizontal component of the error vector is first extracted by setting the vertical component to zero, yielding $\bm{e}_{12} = \begin{bmatrix}
		e_1 & e_2 & 0
	\end{bmatrix}^\top$. This vector is then orthogonalized with respect to $X_\ell$ using the Gram-Schmidt procedure to obtain the $Y_\ell$ as 
	\begin{align}
		\bm{Y}_\ell = \frac{\bm{e}_{12} - (\bm{X}_\ell^\top \bm{e}_{12})\bm{X}_\ell}{\lVert \bm{e}_{12} - (\bm{X}_\ell^\top \bm{e}_{12})\bm{X}_\ell \rVert}.
	\end{align} 
	The third basis vector $Z_l$ is then obtained as the cross product $Z_\ell = X_\ell \times Y_\ell$. Let the $\bm{I}_\ell,$ $\bm{J}_\ell$ and $\bm{K}_\ell$ are unit vector along $X_\ell,$ $Y_\ell$ and $Z_\ell$ axes. The transformation matrix from the inertial frame to the local path frame is constructed as $\bm{T}_s = \begin{bmatrix}
		\bm{I}_\ell^\top & \bm{J}_\ell^\top & \bm{K}_\ell^\top
	\end{bmatrix}^\top$, and the error vector expressed in the local frame is given by 
	\begin{align}
		\bm{e}_\ell = \bm{T}_s \bm{e} = 
		\begin{bmatrix}
			\bm{I}_{\ell,1} & \bm{I}_{\ell,2} & \bm{I}_{\ell,3}\\
			\bm{J}_{\ell,1} & \bm{J}_{\ell,2} & \bm{J}_{\ell,3}\\
			\bm{K}_{\ell,1} & \bm{K}_{\ell,2} & \bm{K}_{\ell,3}
		\end{bmatrix}
		\begin{bmatrix}
			e_1\\ e_2\\ e_3
		\end{bmatrix} = 
		\begin{bmatrix}
			e_{\ell,1}\\ e_{\ell,2}\\ e_{\ell,3}
		\end{bmatrix}.
	\end{align}
	The horizontal and vertical cross-track errors in the local frame are defined as $d_h = \bm{J}_\ell^\top \bm{e} = e_{\ell,2}$ and $d_v = \bm{K}_\ell^\top \bm{e} = e_{\ell,3}$, respectively.
	
	To derive the time derivative of errors, we analyze the relative velocity between the UAV and the local frame. The UAV's velocity vector in the inertial frame is expressed as 
	\begin{align}
		\bm{V}=v
		\begin{bmatrix}
			\cos\gamma\cos\chi &
			\cos\gamma\sin\chi &
			\sin\gamma
		\end{bmatrix}^\top.
	\end{align}
	The desired tangential velocity of the local frame, $\bm{V}_{u}$, at point $P_c$ is obtained by projecting UAV velocity vector along the path tangent vector $\bm{u}$ and its expression can be obtained as $\bm{V}_{u}=\left( \bm{V}^\top \bm{u} \right) \bm{u}.$ The rate of change of the position error vector $\dot{\bm{e}}$ is the difference between UAV inertial velocity and tangential velocity of the local frame, and it is given as $\dot{\bm{e}}=\bm{V} - \bm{V}_u.$ Transforming the rate of change of error $\dot{\bm{e}}$ into the local frame results in $ \dot{\bm{e}}_\ell= \bm{T}_s\dot{\bm{e}}=
	\begin{bmatrix}
		\dot{e}_{l,1} & \dot{e}_{l,2} & \dot{e}_{l,3}        
	\end{bmatrix}^\top$, where $\dot{\bm{e}}_\ell$ is the rate of change of error in the local frame. Consequently, the time derivatives of the horizontal and vertical cross-track errors are obtained as $\dot{d}_h = \dot{e}_{\ell,2}$ and $\dot{d}_v = \dot{e}_{\ell,3}$, respectively. The desired heading angle $\chi_d$ and flight path angle $\gamma_d$ can be obtained as
	\begin{align}
		\chi_{d} = \tan^{-1}\left(\dfrac{X_{\ell,2}}{X_{\ell,1}}\right), \quad \gamma_{d} = \sin^{-1}(X_{\ell,3}).
	\end{align}
	For a straight line path, desired angles $\chi_d$ and $\gamma_d$ are constant, and their time derivatives are zero, that is, $\dot{\chi}_d=0$ and $\dot{\gamma}_d=0$.
	
	We now apply the proposed framework to a 3D circular path, that can be defined as the path starting with base circle of radius $R$ lying in the $X-Y$ plane (normal along $+Z$ axis). The orientation of this base circle is then generalized to an arbitrary spatial configuration by applying two sequential rotations: first, a pitch ($\theta$) about the $y$-axis, and then a roll ($\phi$) about the $x$-axis, both measured positive in the counter-clockwise direction. The resulting orientations transformation matrix, denoted as $\mathcal{R}_{c}$, is given by
	\begin{align}
		\mathcal{R}_{c}&=\mathcal{R}_x(\phi)\mathcal{R}_y(\theta)=
		\begin{bmatrix}
			1 & 0 & 0\\
			0 & \cos\phi & \sin\phi \\
			0 & -\sin\phi & \cos\phi
		\end{bmatrix}
		\begin{bmatrix}
			\cos\theta & 0 & -\sin\theta\\
			0 & 1 & 0\\
			\sin\theta & 0 & \cos\theta
		\end{bmatrix}, \nonumber\\
		&= 
		\begin{bmatrix}
			\cos\theta & 0 & -\sin\theta\\
			\sin\phi\sin\theta & \cos\phi & \sin\phi\cos\theta\\
			\cos\phi\sin\theta & -\sin\phi & \cos\phi\cos\theta
		\end{bmatrix}.
	\end{align}
	Note that the columns of $\mathcal{R}_{c}$ are the orthonormal basis vectors of the circle's local frame. The first column $\bm{u}$ is the direction that corresponds to the $\cos\alpha$ direction, the second column $\bm{v}$ is in the $\sin\alpha$ direction, and the third column $\bm{n}$ is the normal to the circle frame. Using this basis, any point on the circle is denoted as $\mathcal{Z}$, given by
	\begin{align}
		\mathcal{Z}=C+R(\bm{u}\cos\alpha+\bm{v}\sin\alpha),
	\end{align}
	where $C$ is center of circle and $\alpha$ is the angular parameter (from $0$ to $2\pi$). The resulting circular path is shown in \Cref{fig:3D_circular_poblem}. We now compute the closest point on the circular path $P_c$ from the UAV's current position $P_I=\begin{bmatrix}
		x&y&z
	\end{bmatrix}^\top$. Let $\bm{r}_c=P_I-C$ be the vector from the center of the circle to the UAV's current position. This vector is projected onto the circle plane to obtain its components along the basis vectors $\bm{u}$ and $\bm{v}$, obtained as $\bm{u}_{\mathrm{proj}}=\bm{r}_c^\top \bm{u}$ and $ \bm{v}_{\mathrm{proj}}=\bm{r}_c^\top \bm{v}$. The radial direction vector within the circle plane is obtained as
	\begin{align}
		\bm{r}_\parallel=\frac{\bm{u}_{\mathrm{proj}}\bm{u}+\bm{v}_{\mathrm{proj}}\bm{v}}{\sqrt{\bm{u}_{\mathrm{proj}}^2+\bm{v}_{\mathrm{proj}}^2}}   .
	\end{align}
	The tangential component at point $P_c$ is obtained as
	\begin{align}
		\bm{r}_\perp = \frac{\bm{n} \times \bm{r}_\parallel}{||\bm{n} \times \bm{r}_\parallel||}.
	\end{align}
	Thus, The closest point $P_c$ on the circular path is $P_c = C +R\bm{r}_\parallel.$
	
	Similar to straight-line path case, we define a local orthonormal path frame $X_\ell Y_\ell  Z_\ell$ such that $X_\ell$ is $\bm{r}_\perp$ the tangent component in the direction of motion, $Y_\ell$ is $-\bm{r}_\parallel$ in the inward radial direction, and $Z_\ell$ is $\bm{n}$ normal vector to the circle plane. The transformation matrix from the inertial frame to the local frame is $\bm{T}_c=
	\begin{bmatrix}
		\bm{I}_\ell^\top & \bm{J}_\ell^\top & \bm{K}_\ell^\top
	\end{bmatrix}^\top.$
	The position error in the inertial frame is defined as $\bm{e}=P_I-P_c$, then the error in the local frame is $\bm{e}_\ell=\bm{T}_c\,\bm{e} =
	\begin{bmatrix}
		e_{\ell,1} & e_{\ell,2} & e_{\ell,3}
	\end{bmatrix}^\top.$
	The horizontal and vertical cross-track errors $d_h$ and $d_v$ are described as
	\begin{align}
		d_h&=e_{\ell,2} = -\left(\sqrt{\bm{u}_{\mathrm{proj}}^2+\bm{v}_{\mathrm{proj}}^2}-R\right),\\
		d_v&=e_{\ell,3} =\bm{r}_c^\top\bm{n}.
	\end{align}
	For the UAV to follow the circular path, the UAV has to match the tangent velocity vector direction at $P_c$. Hence, the desired  velocity vector $\bm{v}_d$ is given by
	\begin{align}
		\bm{V}_d=v\,\bm{r}_\perp=
		\begin{bmatrix}
			v_{d,1} & v_{d,2} & v_{d,3}
		\end{bmatrix}^\top.
	\end{align}
	The velocity error in the local frame is $\bm{e}_{v,\ell} = \bm{T}_c(\bm{V}-\bm{V}_d)=
	\begin{bmatrix}
		e_{v,\ell,1} & e_{v,\ell,2} & e_{v,\ell,3}
	\end{bmatrix}^\top.$ The time derivative of horizontal and vertical cross-track errors are $\dot{d}_h=e_{v,\ell,2}$ and $\dot{d}_v=e_{v,\ell,3}$, respectively. The desired heading angle $\chi_d $ and flight path angle $\gamma_d$ can be obtained through
	\begin{align}
		\chi_d&=\tan^{-1}\left(\frac{v_{d,2}}{v_{d,1}}\right),~~
		\gamma_d=\sin^{-1}\left(\frac{v_{d,3}}{v}\right).
	\end{align}
	Since the path lies in a plane, $\gamma_d$ is constant  ($\dot{\gamma}_d=0$), while the heading angle is defined by the constant angular velocity of the circular path, $\dot{\chi}_d=\dfrac{v}{R}$.

	\section{Performance Evaluations}\label{sec:simulations}
	The performance of the proposed bounded-input path-following guidance law (\Cref{thm:u}) is demonstrated using numerical simulations for representative paths, namely straight-line, circular, and sinusoidal paths. Unless otherwise stated, the UAV's speed is fixed at $10\, \text{m/s}$, the maximum lateral acceleration bound is set to $a_{m\max} = M_2' = 10\ \text{m/s}^2$, with $M_1 = M_2/2.1$. The design parameters are chosen as $k_1 = k_2 = 1$. 
	
	We first consider a straight-line path from $P_1=[0,~0]^\top$ to $P_2=[200,~200]^\top$, as shown in \Cref{fig:2a}. Four initial conditions are chosen as: $(x,y,\psi)=(20,\,10,\,45^\circ)$, $(10,\,30,\,80^\circ)$, $(15,\,-15,\,135^\circ)$, and $(0,\,40,\,60^\circ)$. Under the proposed guidance law, the performance of the UAV is demonstrated through \Cref{fig:3}. As evident from \Cref{fig:3a}, the UAV converges to the desired straight-line path in all cases. The UAV's lateral acceleration remains within its prescribed bounds $\pm 10$ m/s$^2$ (\Cref{fig:3b}). One may notice from \Cref{fig:3d} that once both $\sigma_1(.)$ and $\sigma_2(.)$ become unsaturated simultaneously, the cross-track error $d$ and its derivative $\dot{d}$ exhibit exponential convergence to zero (see \Cref{fig:3c}).
	\begin{figure*}[!ht]
		\centering
		\begin{subfigure}{0.25\linewidth}
			\centering
			\includegraphics[width=\linewidth]{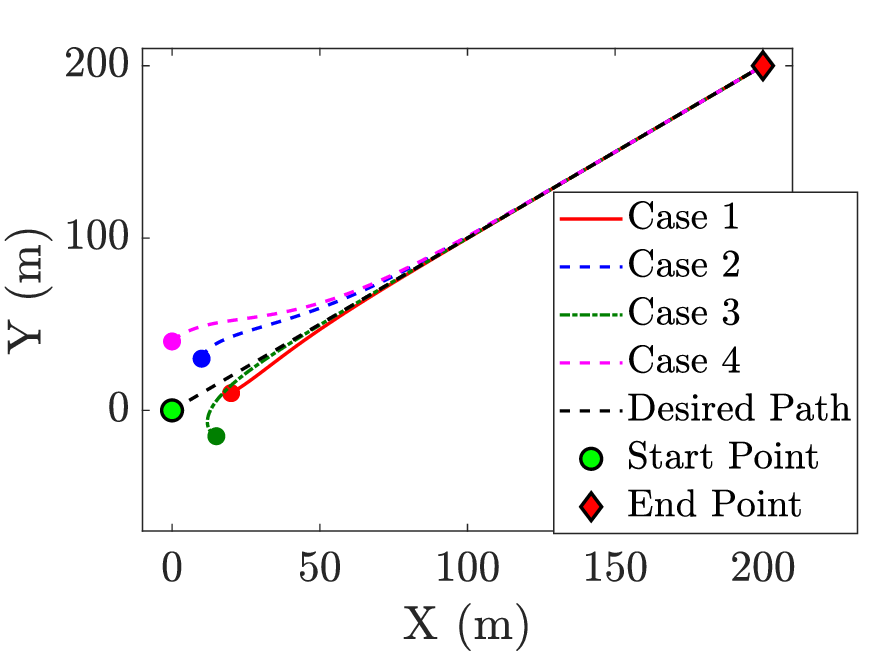}
			\caption{Trajectories.}
			\label{fig:3a}    
		\end{subfigure}%
		\begin{subfigure}{0.25\linewidth}
			\centering
			\includegraphics[width=\linewidth]{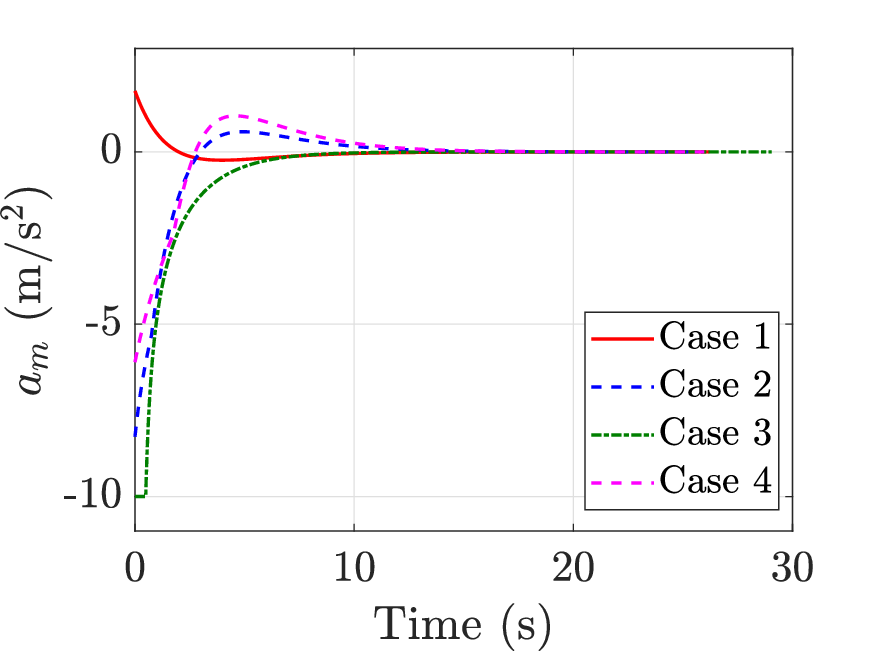}
			\caption{Lateral accelerations.}
			\label{fig:3b}    
		\end{subfigure}%
		\begin{subfigure}{0.25\linewidth}
			\centering
			\includegraphics[width=\linewidth]{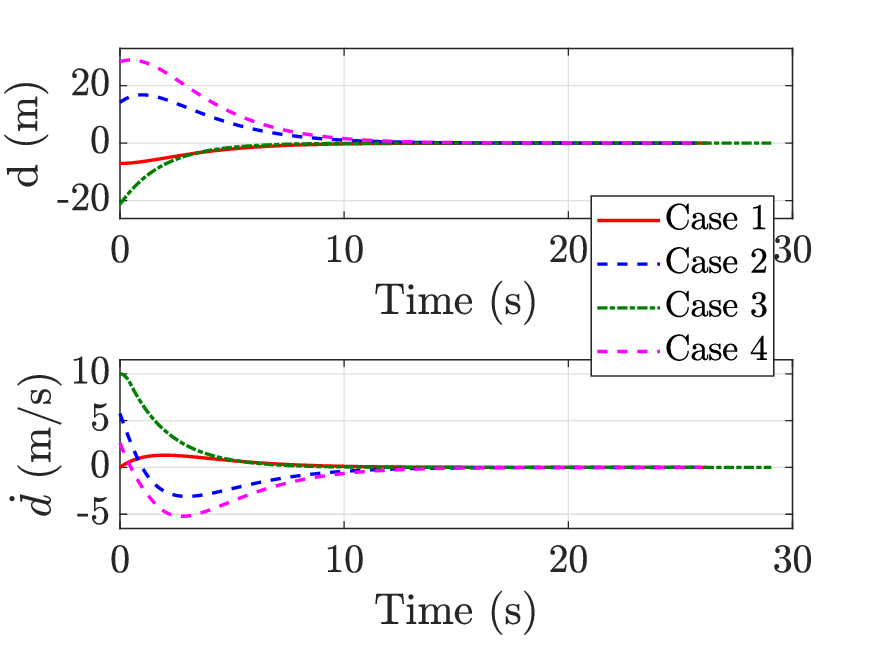}
			\caption{Error and its rate.}
			\label{fig:3c}    
		\end{subfigure}%
		\begin{subfigure}{0.25\linewidth}
			\centering
			\includegraphics[width=\linewidth]{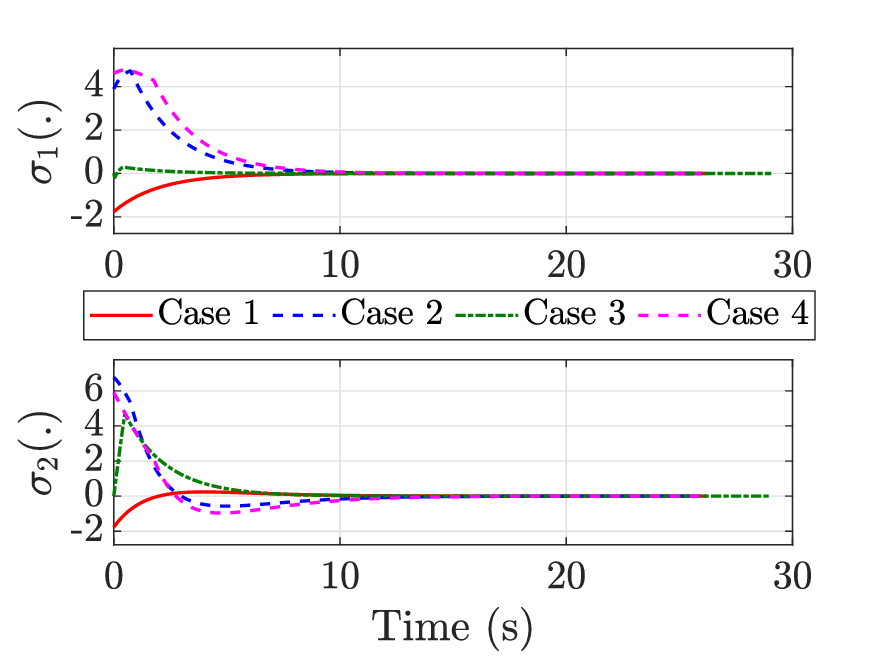}
			\caption{Saturation functions.}
			\label{fig:3d}    
		\end{subfigure}
		\caption{UAV's performance: straight-line path.}
		\label{fig:3}
	\end{figure*}
	
	Next, we consider a circular path of radius $20$ m as the desired path, shown in \Cref{fig:2b}. The performance of the UAV is evaluated under four distinct initial conditions: $(x,y,\psi)=(20,\,20,\,60^\circ)$, $(10,\,10,\,45^\circ)$, $(30,\,20,\,100^\circ)$, and $(-25,-40,130^\circ)$. Under the proposed guidance law \eqref{eq:input}, we illustrate the UAV's performance in \Cref{fig:4}. Unlike the straight-line scenario, the circular path requires a constant curvature to be maintained. Despite of this, the UAV is able to converge to the circular path from all four initial conditions (\Cref{fig:4a}), while respecting the bounded control
	inputs of $\pm 10$ m/s$^2$. Similar to the straight-line case, the cross-track error $d$ and its derivative $\dot{d}$ converge to zero for all cases, as illustrated in \Cref{fig:4c}, ensuring exponential convergence to the circular path.
	\begin{figure*}[!ht]
		\centering
		\begin{subfigure}{0.25\linewidth}
			\centering
			\includegraphics[width=\linewidth]{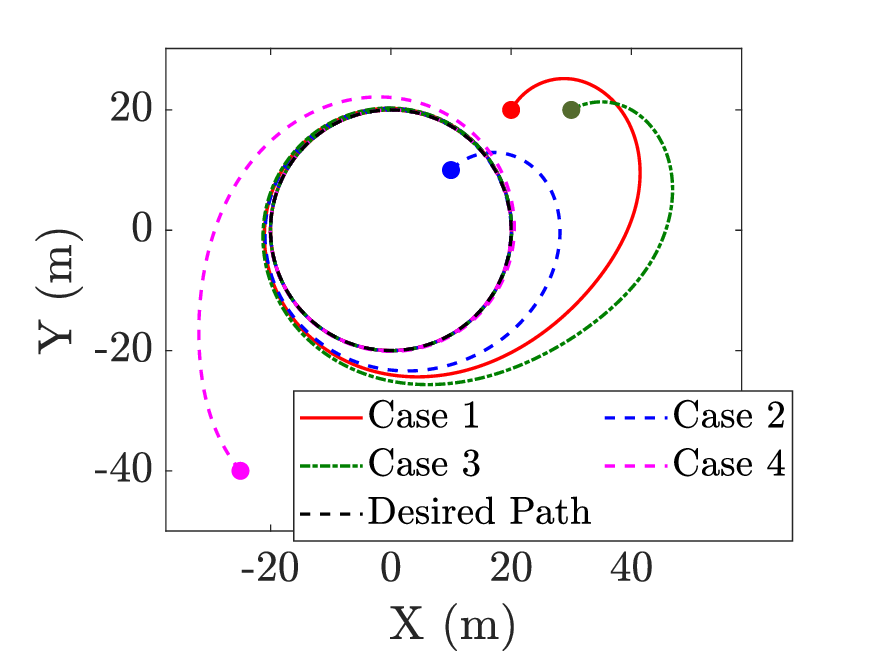}
			\caption{Trajectories.}
			\label{fig:4a}    
		\end{subfigure}%
		\begin{subfigure}{0.25\linewidth}
			\centering
			\includegraphics[width=\linewidth]{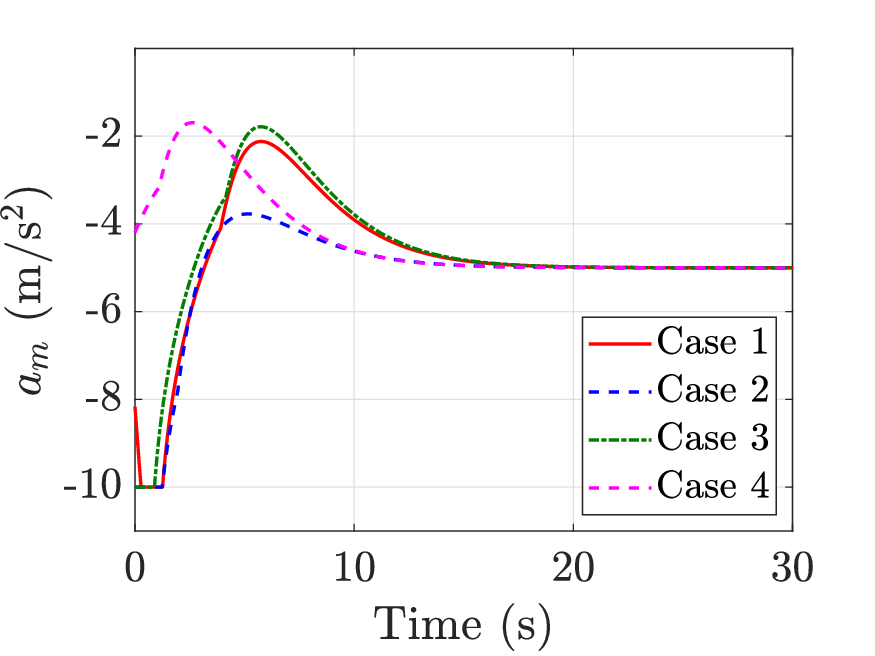}
			\caption{Lateral accelerations.}
			\label{fig:4b}    
		\end{subfigure}%
		\begin{subfigure}{0.25\linewidth}
			\centering
			\includegraphics[width=\linewidth]{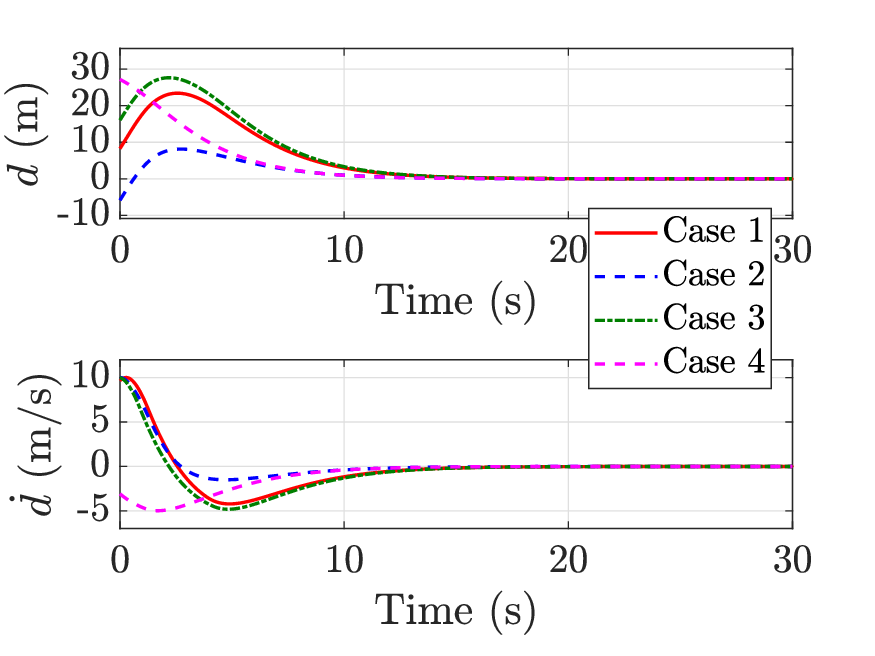}
			\caption{Error and its rate.}
			\label{fig:4c}    
		\end{subfigure}%
		\begin{subfigure}{0.25\linewidth}
			\centering
			\includegraphics[width=\linewidth]{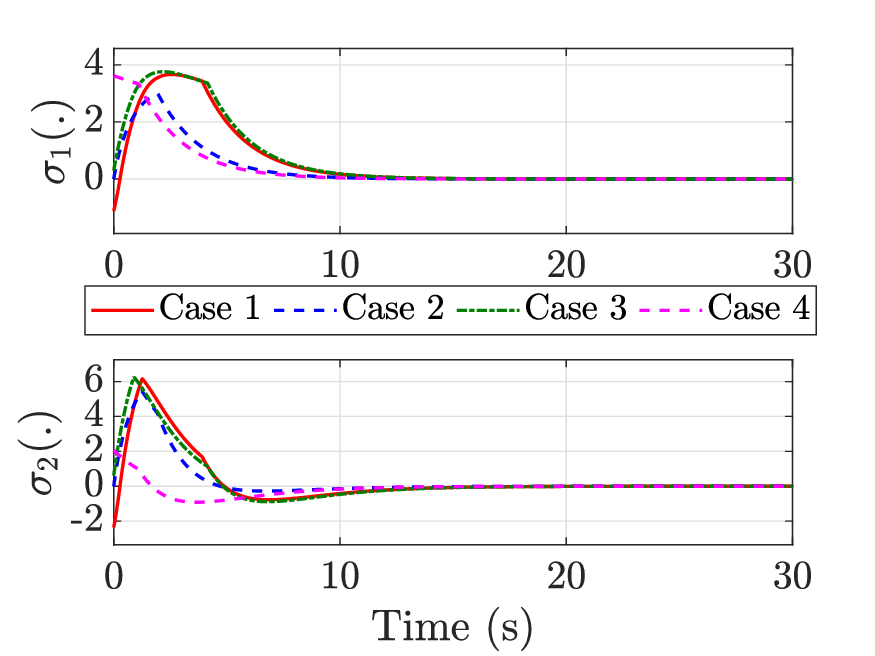}
			\caption{Saturation functions.}
			\label{fig:4d}    
		\end{subfigure}
		\caption{UAV's performance: following a circular path.}
		\label{fig:4}
	\end{figure*}
	
	We now demonstrate the performance of the proposed guidance law on the sinusoidal curve to show that our proposed controller converges to any arbitrary smooth path under input constraint. For this case, two distinct initial conditions are considered: $(x, y, \psi) = (15, 15, 30^\circ)$ and $(10, -10, 60^\circ)$. The sinusoidal curve parameters are chosen as $A=10$ m and $\omega=\frac{2\pi}{100}$ rad/s, respectively. The UAV performance under the proposed guidance law \eqref{eq:input} is illustrated in \Cref{fig:5}. The sinusoidal path introduces a continuously varying curvature, making it a more challenging test than either straight-line or circular cases. However, it is observed that the guidance law successfully steers the UAV onto the desired sinusoidal path by exhibiting almost similar behavior as for the straight-line and circular path.
	
	\begin{figure*}[!ht]
		\centering
		\begin{subfigure}{0.25\linewidth}
			\centering
			\includegraphics[width=\linewidth]{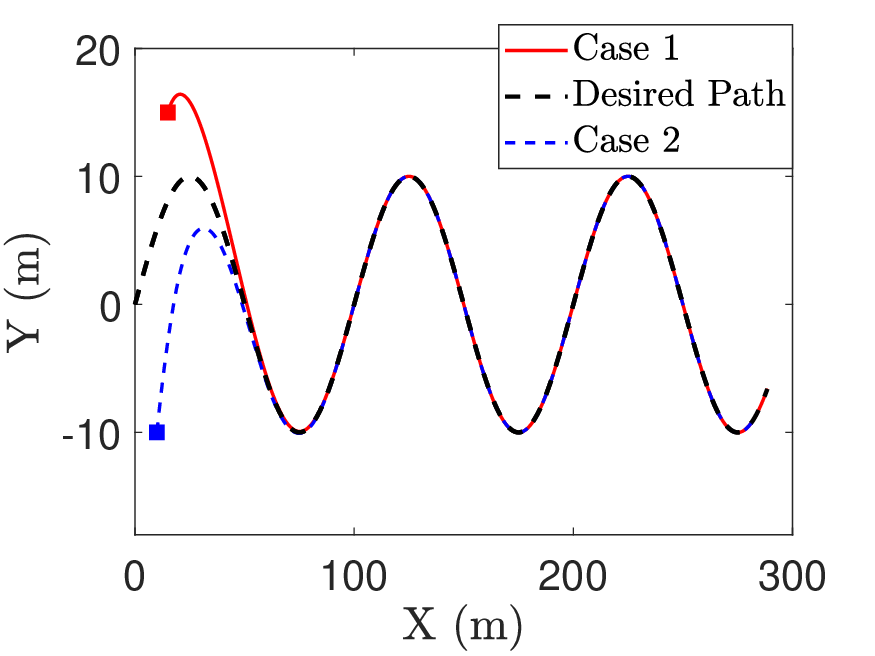}
			\caption{Trajectories.}
			\label{fig:5a}    
		\end{subfigure}%
		\begin{subfigure}{0.25\linewidth}
			\centering
			\includegraphics[width=\linewidth]{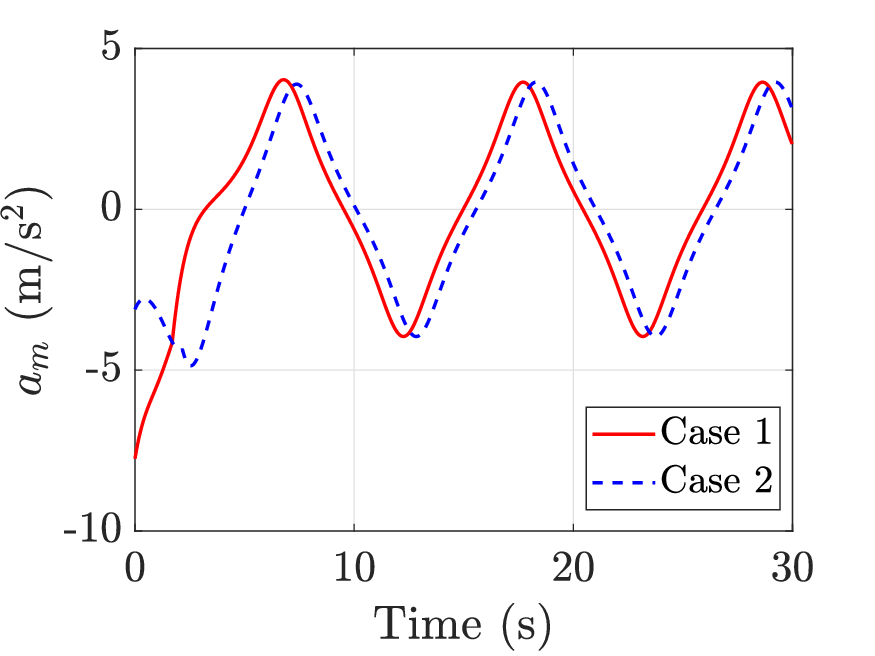}
			\caption{Lateral accelerations.}
			\label{fig:5b}    
		\end{subfigure}%
		\begin{subfigure}{0.25\linewidth}
			\centering
			\includegraphics[width=\linewidth]{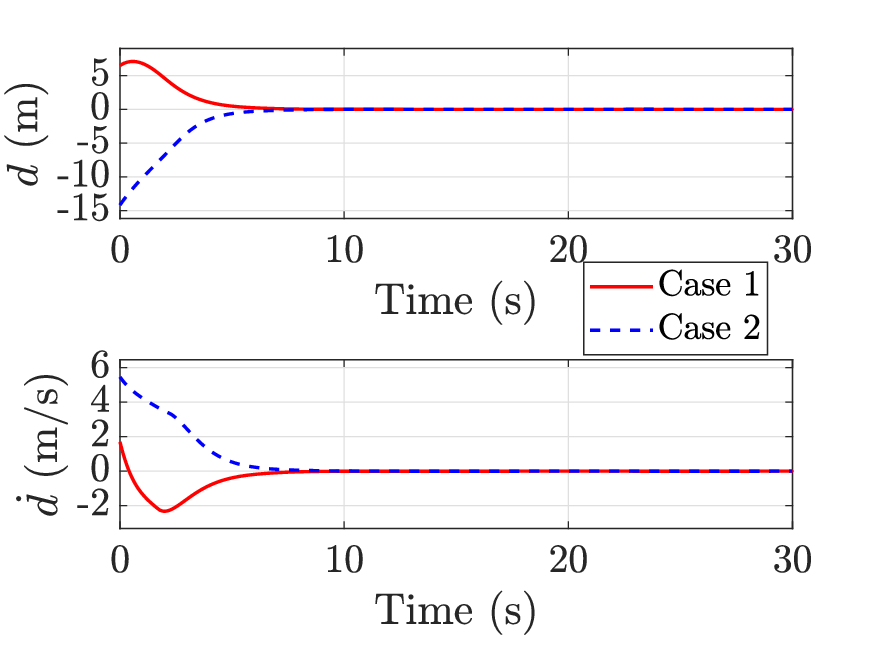}
			\caption{Error and its rate.}
			\label{fig:5c}    
		\end{subfigure}%
		\begin{subfigure}{0.25\linewidth}
			\centering
			\includegraphics[width=\linewidth]{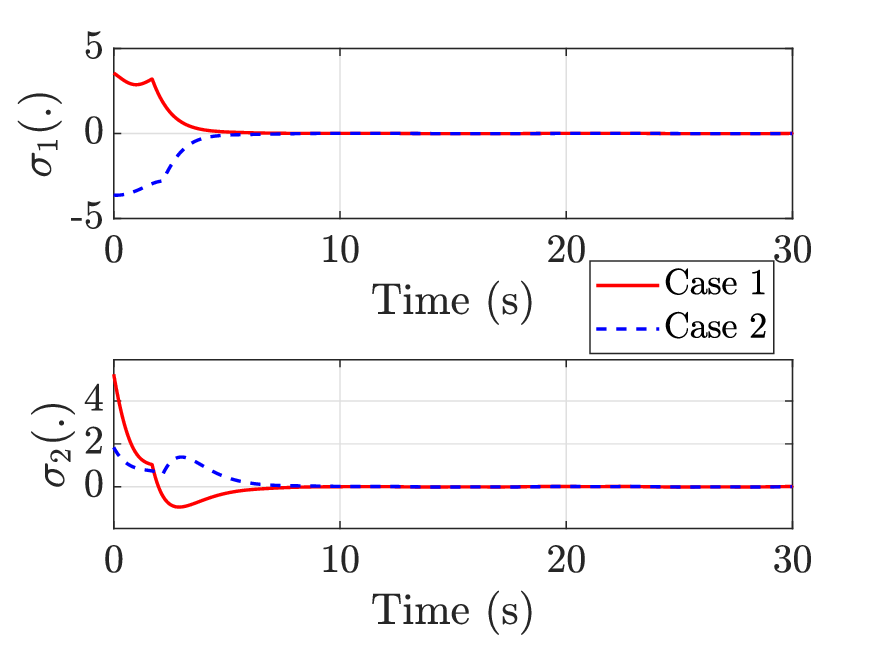}
			\caption{Saturation functions.}
			\label{fig:5d}    
		\end{subfigure}
		\caption{UAV's performance: following a sinusoidal path.}
		\label{fig:5}
	\end{figure*}
	
	We now demonstrate the robustness of the proposed guidance law. For that, we consider a scenario in which the UAV encounters a wind gust while following a circular path. The objective is to verify that the UAV can successfully return to its desired path after the wind gust ends, while maintaining bounded control input. The simulation parameters are chosen as follows: the radius of the circular path is 50 m, and the UAV initial conditions are $x=100$ m, $y=100$ m, $\psi=45^\circ$. A wind gust is characterized by a velocity of $5\sqrt{2}$ \si{m/s} at an angle of $45^\circ$ with reference to the x-axis, and it is active for 10 seconds, as illustrated in \Cref{fig:2D_gust_d}. Under the action of the proposed guidance law, the performance of the UAV is shown in \Cref{fig:2D_gust}. From \Cref{fig:2D_gust}, we observed that the UAV successfully converged to its desired path in the presence of the wind gust while maintaining bounded control input. As shown in \Cref{fig:2D_gust_c}, both cross-track error $d$ and its time derivative $\dot{d}$ converge to zero after the wind gust ends. This convergence behavior confirms that the proposed guidance law exhibits a robust performance in the presence of external disturbances.
	\begin{figure*}[!ht]
		\centering
		\begin{subfigure}{0.25\linewidth}
			\centering
			\includegraphics[width=\linewidth]{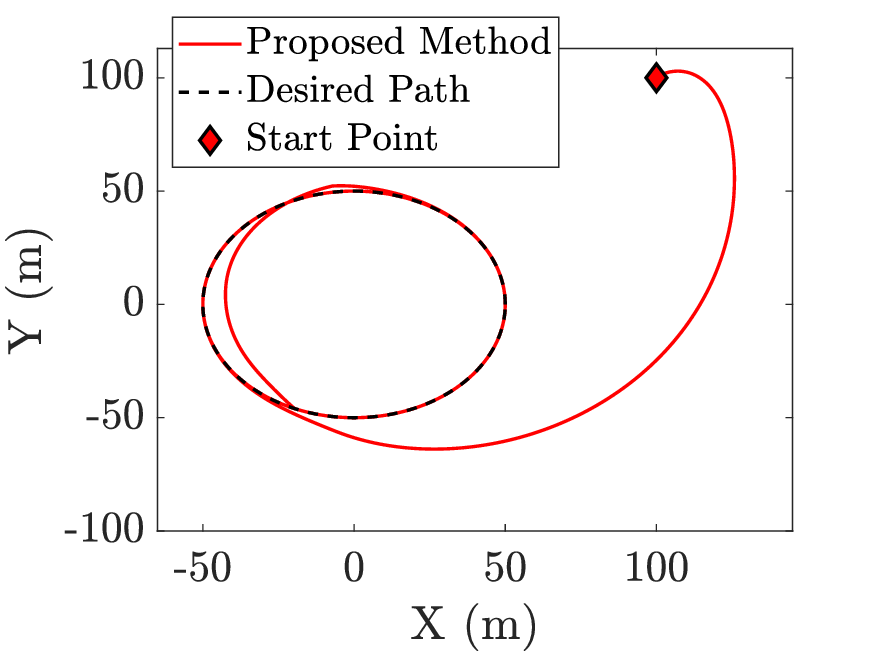}
			\caption{Trajectories.}
			\label{fig:2D_gust_a}    
		\end{subfigure}%
		\begin{subfigure}{0.25\linewidth}
			\centering
			\includegraphics[width=\linewidth]{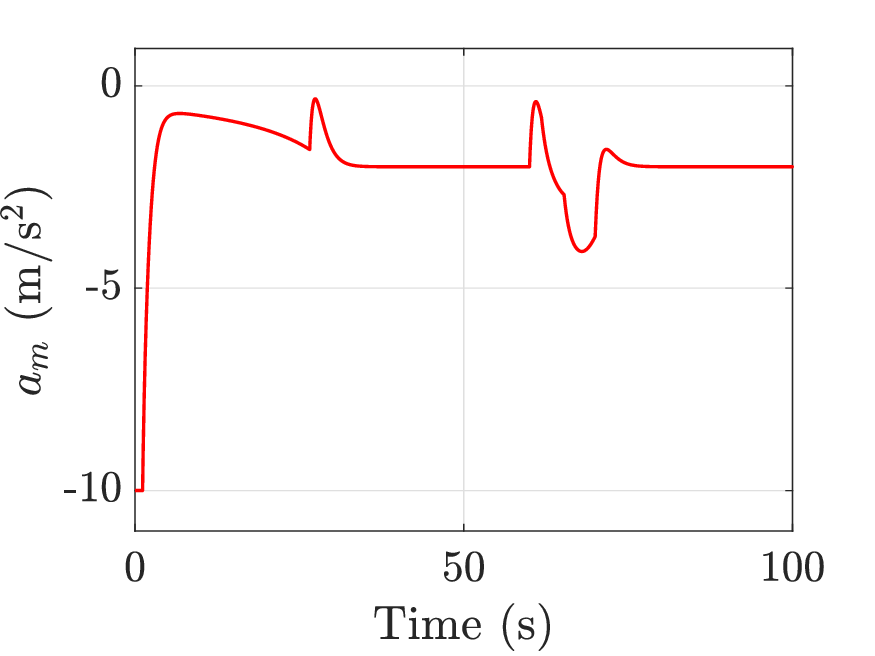}
			\caption{Lateral accelerations.}
			\label{fig:2D_gust_b}    
		\end{subfigure}%
		\begin{subfigure}{0.25\linewidth}
			\centering
			\includegraphics[width=\linewidth]{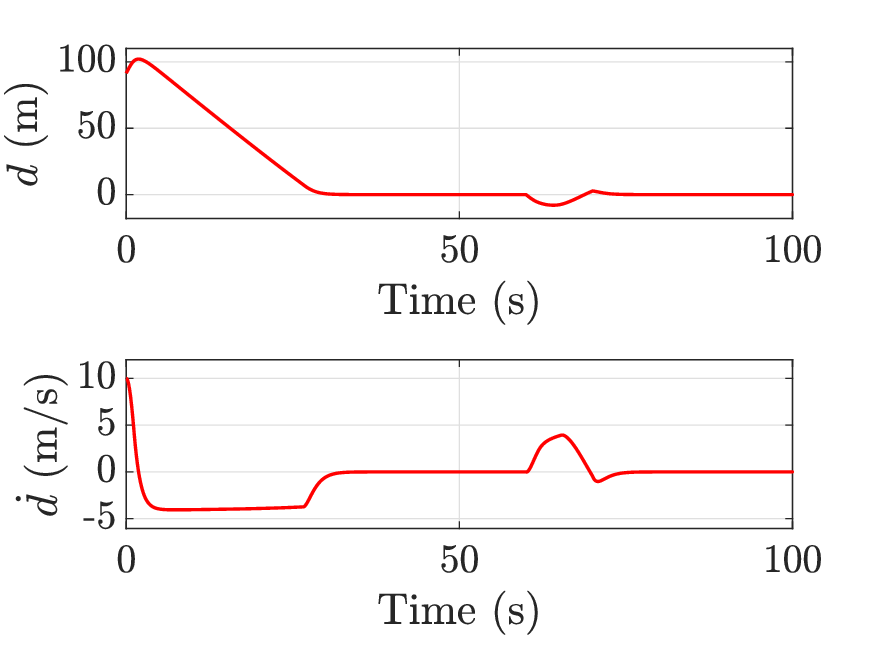}
			\caption{Error and its rate.}
			\label{fig:2D_gust_c}    
		\end{subfigure}%
		\begin{subfigure}{0.25\linewidth}
			\centering
			\includegraphics[width=\linewidth]{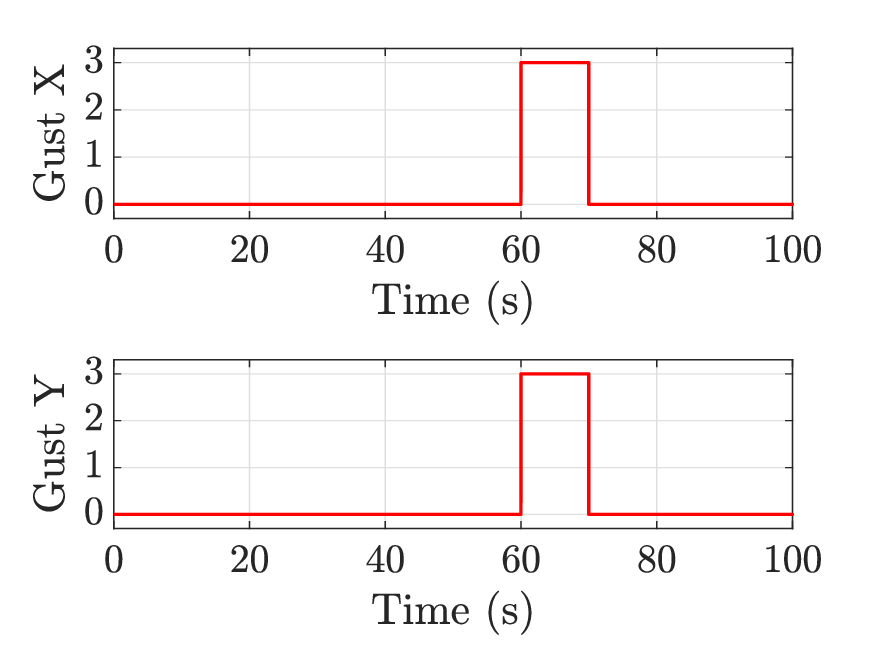}
			\caption{Saturation functions.}
			\label{fig:2D_gust_d}    
		\end{subfigure}
		\caption{UAV's performance: in the presence of gust.}
		\label{fig:2D_gust}
	\end{figure*}
	
	It is worth highlighting that the proposed constrained input path-following guidance laws drive the UAV onto the desired path, regardless of its initial engagement geometry with respect to the path. This, in turn, supports the claim that the proposed guidance strategy ensures global convergence of the UAV to the desired path. Moreover, the proposed strategy guides the UAV onto a desired path (smoothly), which supports the claim that the guidance strategy is generic in nature and will remain applicable to any smooth path. The sinusoidal path can serve as the basis for smooth paths.
	
	Next, we contrast the performance of the proposed bounded input path following guidance law with existing path following methods presented in \cite{ratnoo2011adaptive}, \cite{kothari2014uav}, \cite{kumar2023robust}, and \cite{patrikar2019nested}, denoted by $\mathcal{C}_1$, $\mathcal{C}_2$, $\mathcal{C}_3$, and $\mathcal{C}_4$, respectively.  
	
	In \cite{ratnoo2011adaptive}, the guidance law was developed based on an adaptive optimal technique to manipulate the heading angle and to make the UAV converge on the desired path. The controller of $\mathcal{C}_1$ is given by 
	\begin{align*}
		a_m=-\left[\left(\sqrt{\left|\frac{d_b}{d_b-d}\right|}\right)d+\left(\sqrt{2\sqrt{\left|\frac{d_b}{d_b-d}\right|}+1}\right)\dot{d}\right],
	\end{align*}
	where $d_b$ is maximum desired distance error band and $d$ is cross-track error. On the other hand, the authors of $\mathcal{C}_2$ developed pure pursuit and line-of-sight-based guidance laws to steer the UAV on the desired path. The guidance law of Method 2 is $a_m=a_1(\psi_d-\psi)+a_2 d$, where $a_1$ and $a_2$ are positive gains, $d$ is cross-track error, and $\psi_d$ denotes the desired heading angle. Note that $\mathcal{C}_2$ is developed only for special paths like straight-line and circular orbit paths. So, we restrict ourselves to these generic paths. The author of $\mathcal{C}_3$ used SMC to drive $d$ and $\dot{d}$ to zero. The sliding mode controller of $\mathcal{C}_3$ is 
	\begin{align*}
		a_m&=\frac{-1}{\cos(\psi-\psi_d)}\left[\beta\frac{p}{q}\dot{d}^{2-\frac{p}{q}}+\eta \sign(s)\right],s=d+\frac{1}{\beta}\dot{d}^{\frac{p}{q}},
	\end{align*}
	where $\beta>0$, $p$ and $q$ are odd positive integers such that $p>q$, and $1< \frac{p}{q}<2$, and $\eta$ is a positive constant. The author of $\mathcal{C}_4$ used nested saturation theory to design a guidance law for path following under bounded input, which is given by
	\begin{align*}
		a_m&=-\sigma_{H_1}\left(\frac{s_1\dot{d}+\sigma_{H_2}(s_2s_1\dot{d}+s_2d)}{\cos(\psi-\psi_d)}\right),\\
		\sigma_{H_i}(r)&=
		\begin{cases}
			H_i,  & r\ge H_i\\
			-H_i, & r\le -H_i \\
			r, & \text{otherwise}
		\end{cases}.
	\end{align*}
	The terms $H_1$, $H_2$, $s_1$, and $s_2$ are chosen to satisfy the following conditions:
	\begin{align*}
		\frac{H_2}{s_1} &< v_m, \\
		\frac{H_2}{\sigma_{H_1}\left(\frac{H_2+\sigma_{H_2}(s_1H_2)}{\sqrt{1-{H_2^2}/{\left(v_m^2s_1^2\right)}}}\right) \sqrt{1-{H_2^2}/{\left(v_m^2s_1^2\right)}}} &< \frac{s_1+s_1s_2}{s_2},
	\end{align*}
	where $H_1$ and $H_2$ are bounds on $\sigma_{H_1}(.)$ and $\sigma_{H_2}(.)$, $s_1$ and $s_2$ are gains, and $v_m$ is the UAV's speed, respectively.
	
	We first contrast the performance for straight-line path following. The UAV is initially located at $(10~\text{m}, -100~\text{m})$ with a heading angle of $90^\circ$. The UAV velocity is $10~\text{m/s}$, and it is required to follow a straight-line path whose start and end points, $P_1$ and $P_2$, are $(0~\text{m}, 0~\text{m})$ and $(300~\text{m}, 300~\text{m})$, respectively. For the simulation purpose, the parameters of the proposed method were selected as follows: maximum commanded lateral acceleration $M_2 = 10~\text{m/s}^2$, and gains $k_1 = 0.2$, $k_2 = 0.2$. For $\mathcal{C}_1$, the maximum cross-track error band is $d_b = 5$. The gains of $\mathcal{C}_2$ are $a_1 = 30$ and $a_2 = 1$. Sliding mode parameters of $\mathcal{C}_3$ are $\beta = 5$, $\eta = 15$, $p = 15$, and $q = 13$. The parameters for $\mathcal{C}_4$ are $H_2 = 9$, $H_1 = 10$, $s_1 = 1.5$, and $s_2 = 4$. Simulation results for these guidance laws are shown in \Cref{fig:6}. The guidance law of the proposed method steers the UAV to the desired straight-line path while maintaining bounded control input. In contrast, $\mathcal{C}_1$, $\mathcal{C}_2$, $\mathcal{C}_3$, and $\mathcal{C}_4$ converge to the desired path, but the control inputs do not remain within the desired bounds. The root mean square (RMS) control effort of the proposed method, $\mathcal{C}_1$, $\mathcal{C}_2$, $\mathcal{C}_3$, and $\mathcal{C}_4$, is $0.2616$, $2.8346$, $3.1879$, $1.0511$, and $0.5103$, respectively. The RMS control effort of the proposed method is significantly lower than other methods, indicating that it is most suitable for long-duration maneuvers as it consumes less power. The guidance law of $\mathcal{C}_4$ also converges the UAV to the desired path while maintaining bounded input, but its RMS control effort is nearly double times high as that of the proposed method. Note that $\mathcal{C}_1$, $\mathcal{C}_2$, and $\mathcal{C}_3$ exhibit high lateral accelerations, which are not well-suited for real-world applications and may damage UAV systems.
	
	Finally, we compare the performance when the UAV is subjected to a circular orbit path. The initial position of the UAV is $(100~\text{m}, 50~\text{m})$ with a heading angle of $45^\circ$. It is moving with a constant speed of $10~\text{m/s}$ and has to follow a circular orbital path of radius $50~\text{m}$ with center at $(0~\text{m}, 0~\text{m})$. For this case, the parameters for method 2 are $a_1 = 30$ and $a_2 = 0.1$. Sliding mode parameters of method 3 are $\beta = 3$, $\eta = 10$, $p = 5$, and $q = 3$. The parameters for method 4 are $H_2 = 9$, $H_1 = 10$, $s_1 = 1.5$, and $s_2 = 4$. The other parameters are kept the same as for the straight-line case. The comparative simulation results are illustrated in \Cref{fig:7}. The controller of the proposed strategy smoothly converges to the desired circular path with minimal cross-track error as compared to $\mathcal{C}_1$, $\mathcal{C}_2$, $\mathcal{C}_3$, and $\mathcal{C}_4$, as shown in \Cref{fig:7c}. The RMS control effort of the proposed method, $\mathcal{C}_1$, $\mathcal{C}_2$, $\mathcal{C}_3$, and $\mathcal{C}_4$ are 2.4163, 3.5215, 4.8541, 3.3005, and 2.4893, respectively. The rate of change of error for all methods converges to zero except $\mathcal{C}_2$.  
	
	Note that the proposed guidance strategy demonstrates superior performance in terms of path-following accuracy, input boundedness, and total control effort, making it a more reliable and robust solution for real-time UAV applications.
	\begin{figure*}[!ht]
		\centering
		\begin{subfigure}{0.33\linewidth}
			\centering
			\includegraphics[width=\linewidth]{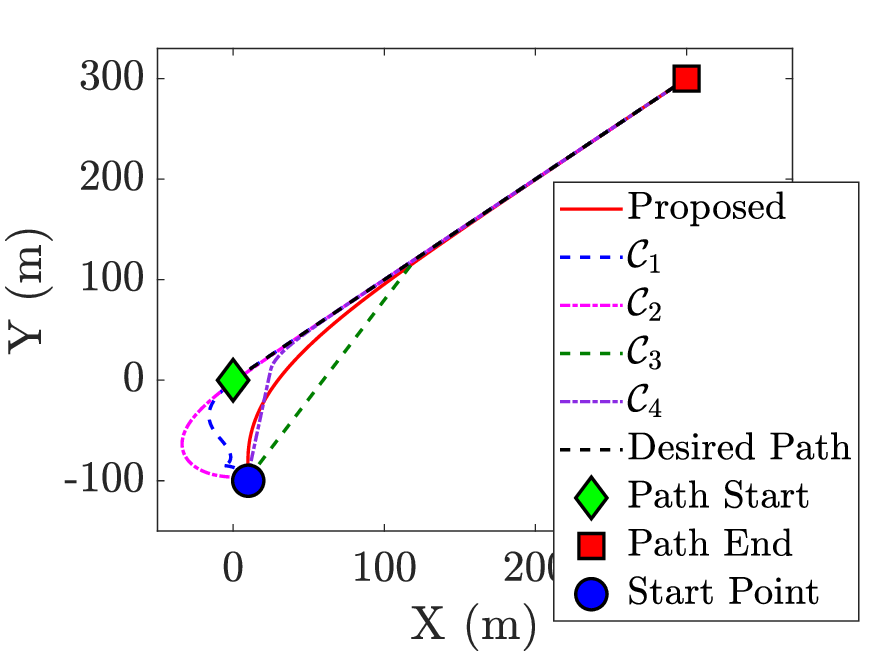}
			\caption{Trajectories.}
			\label{fig:6a}    
		\end{subfigure}%
		\begin{subfigure}{0.33\linewidth}
			\centering
			\includegraphics[width=\linewidth]{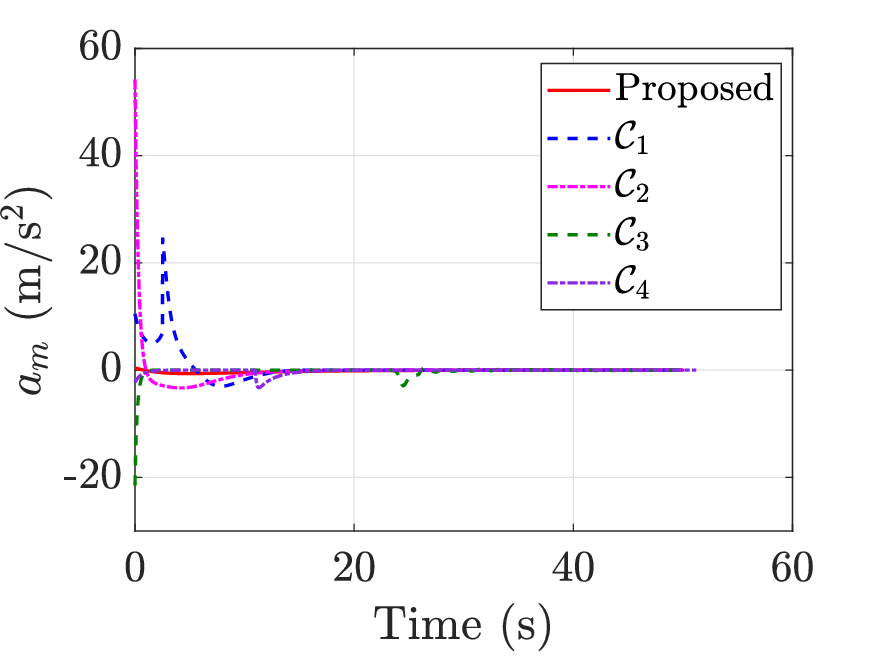}
			\caption{Lateral accelerations.}
			\label{fig:6b}    
		\end{subfigure}%
		\begin{subfigure}{0.33\linewidth}
			\centering
			\includegraphics[width=\linewidth]{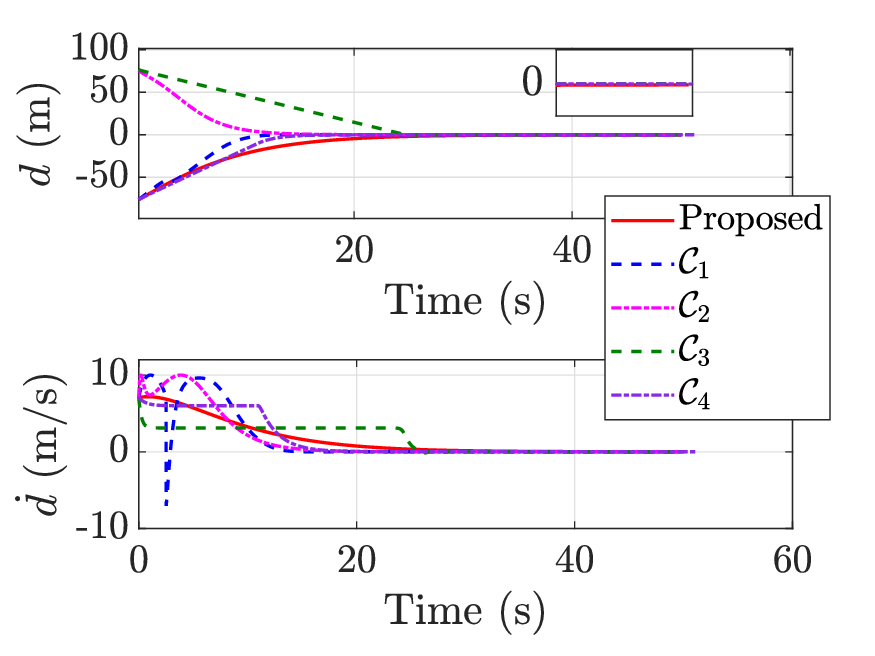}
			\caption{Error and its rate.}
			\label{fig:6c}    
		\end{subfigure}
		\caption{Performance comparison: straight-line path}
		\label{fig:6}
	\end{figure*}
	\begin{figure*}[!ht]
		\centering
		\begin{subfigure}{0.33\linewidth}
			\centering
			\includegraphics[width=\linewidth]{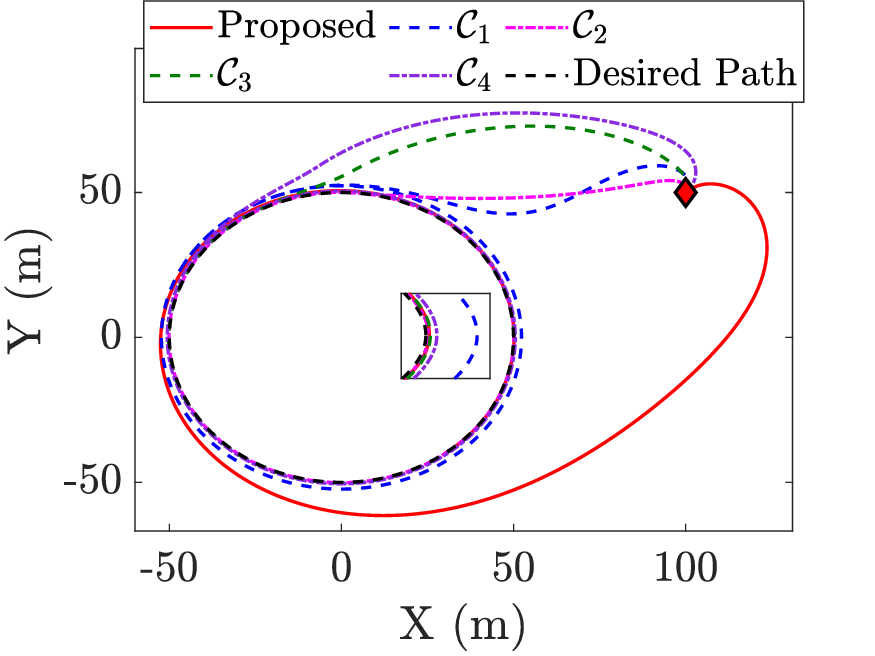}
			\caption{Trajectories.}
			\label{fig:7a}    
		\end{subfigure}%
		\begin{subfigure}{0.33\linewidth}
			\centering
			\includegraphics[width=\linewidth]{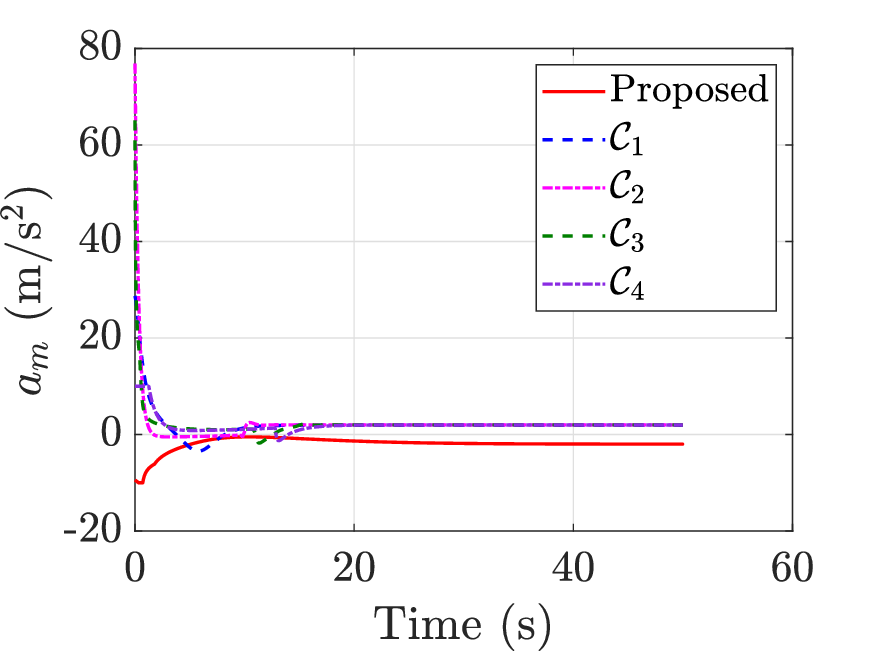}
			\caption{Lateral accelerations.}
			\label{fig:7b}    
		\end{subfigure}%
		\begin{subfigure}{0.33\linewidth}
			\centering
			\includegraphics[width=\linewidth]{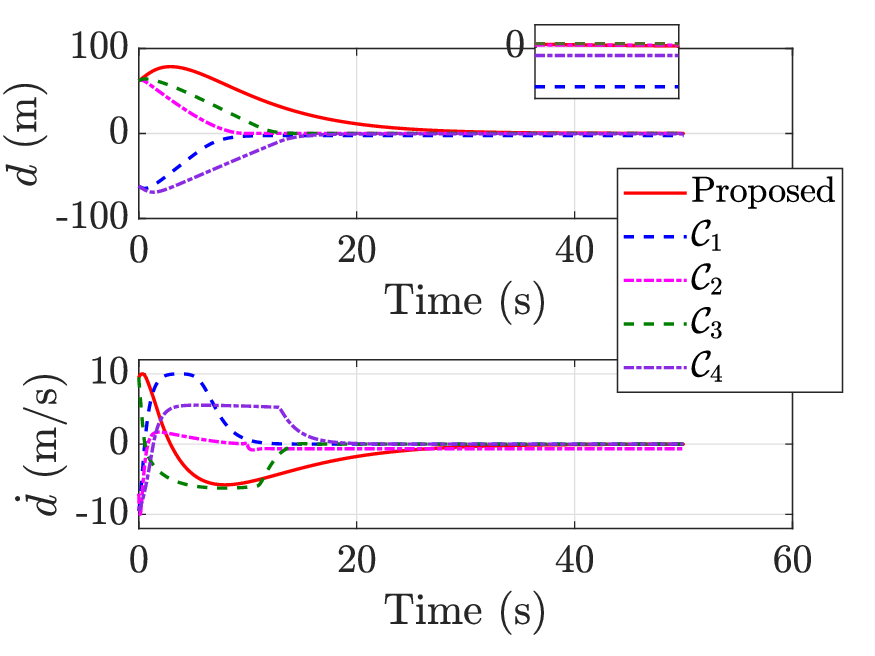}
			\caption{Error and its rate.}
			\label{fig:7c}    
		\end{subfigure}
		\caption{Performance comparison: circular path}
		\label{fig:7}
	\end{figure*}
	
	We now consider the 3D straight-line path from point $P_3=[0,0,0]^\top$ to $P_4=[200,200,200]$, as shown in \Cref{fig:3D_straightline_poblem}. The simulation parameters are chosen as follows: UAV speed of $15$ m/s, maximum acceleration bounds $a_{m\max} = M_{2,h}'=M_{2,v}'=10$ m/s$^2$, and gains chosen as $k_{1,h}=k_{2,h}=k_{1,v}=k_{2,v}=1$. To evaluate robustness against initial conditions, three distinct initial conditions are chosen as: $(x,y,z,\chi,\gamma) = (10, 20, 30, 30^\circ, 40^\circ)$, $(20, 40, -60, 40^\circ, 20^\circ)$, and $(50, 70, 80, 50^\circ, 20^\circ)$. Under the proposed guidance law, the performance of the UAV is demonstrated through \Cref{fig:11}. As observed, the UAV converges to the desired 3D straight-line path in all cases, driving the horizontal and vertical states $d_h, d_v, \dot{d}_h, \dot{d}_v$ to zero (see \Cref{fig:11c}, \Cref{fig:11d}). The control inputs $a_h$ and $a_v$ remain strictly bounded within the prescribed bounds of $\pm 10$ m/s$^2$ for the entire duration $t \geq 0$ (see \Cref{fig:11b}). Similarly to the planar straight-line case, once the saturation functions $\sigma_{1,h}(\cdot)$ and $\sigma_{2,h}(\cdot)$ become unsaturated simultaneously (see \Cref{fig:11e}, \Cref{fig:11f}), the horizontal cross-track error $d_h$ and its time derivative $\dot{d}_h$ exhibit an exponential convergence to zero (see \Cref{fig:11c}, \Cref{fig:11d}). The same behavior can also be observed in the vertical plane.   
	
	\begin{figure*}[htbp]
		\centering
		\begin{subfigure}{0.33\linewidth}
			\centering
			\includegraphics[width=\linewidth]{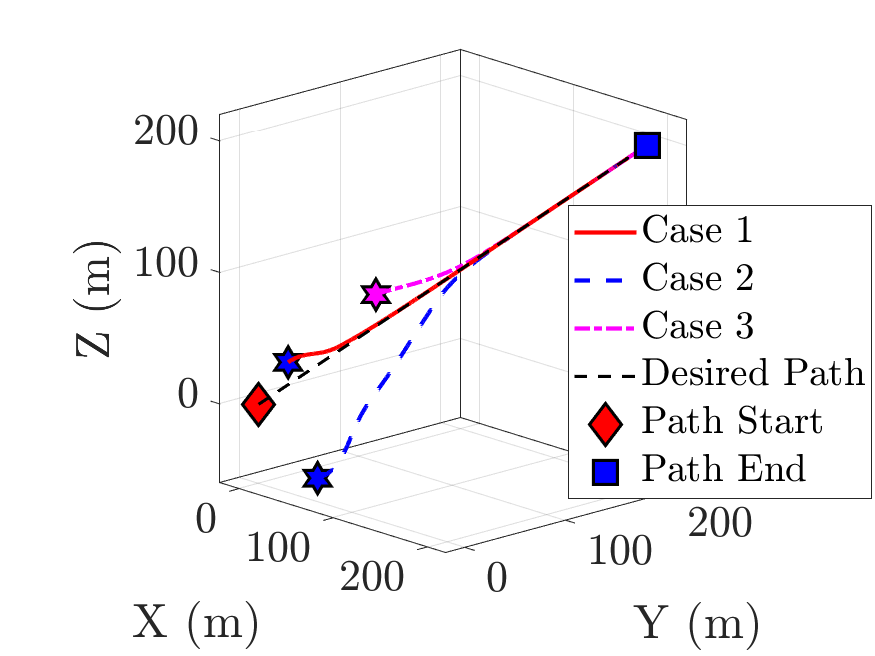}
			\caption{Trajectories.}
			\label{fig:11a}    
		\end{subfigure}%
		\begin{subfigure}{0.33\linewidth}
			\centering
			\includegraphics[width=\linewidth]{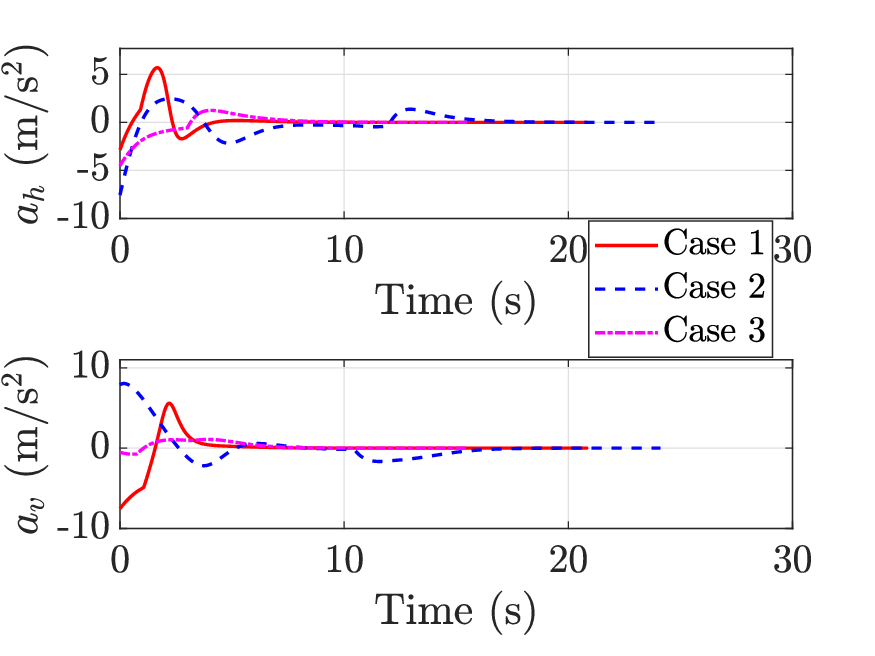}
			\caption{Lateral Accelerations. }
			\label{fig:11b}    
		\end{subfigure}%
		\begin{subfigure}{0.33\linewidth}
			\centering
			\includegraphics[width=\linewidth]{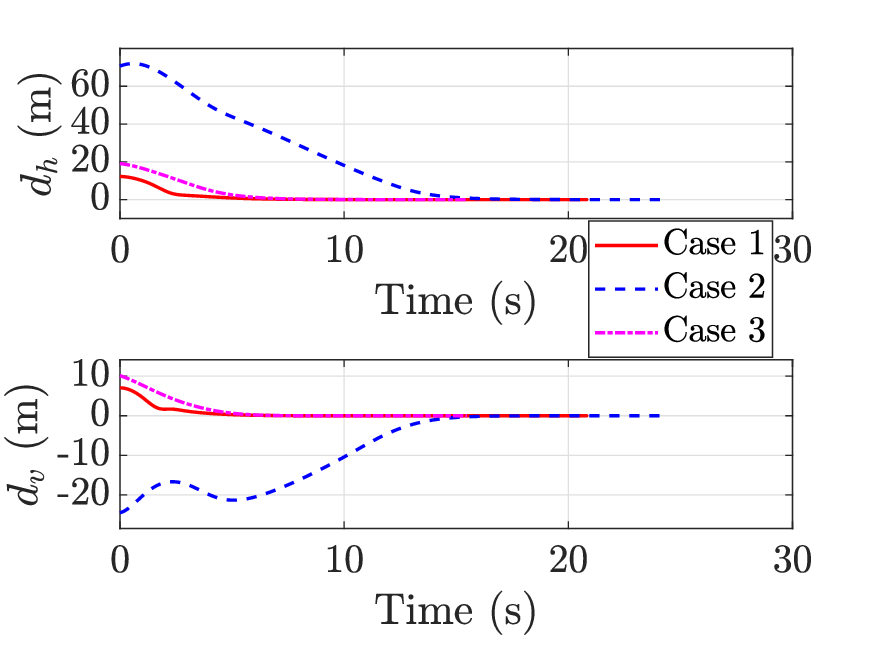}
			\caption{Distance Error.}
			\label{fig:11c}
		\end{subfigure}
		\begin{subfigure}{0.33\linewidth}
			\centering
			\includegraphics[width=\linewidth]{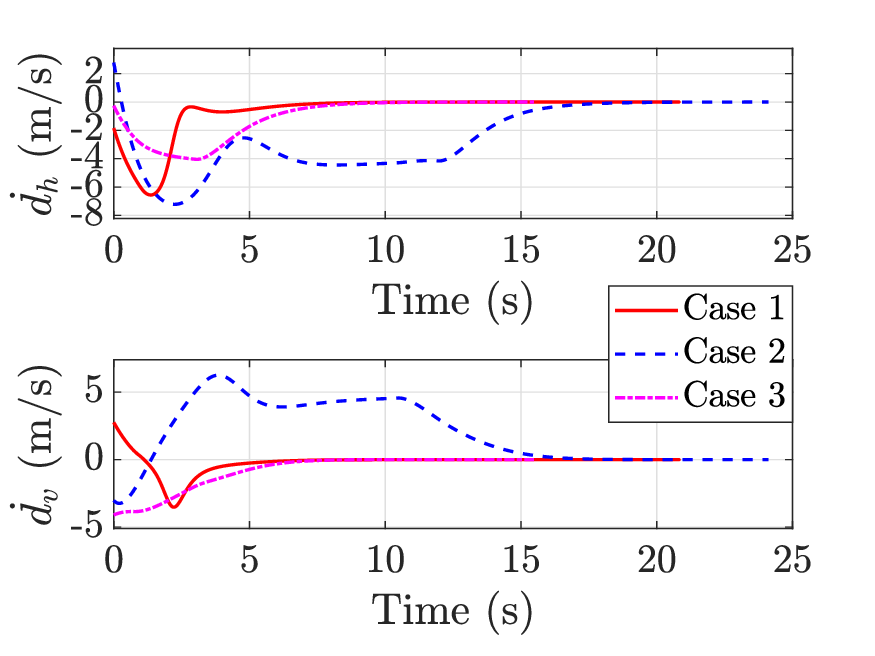}
			\caption{Velocity Error.}
			\label{fig:11d}    
		\end{subfigure}%
		\begin{subfigure}{0.33\linewidth}
			\centering
			\includegraphics[width=\linewidth]{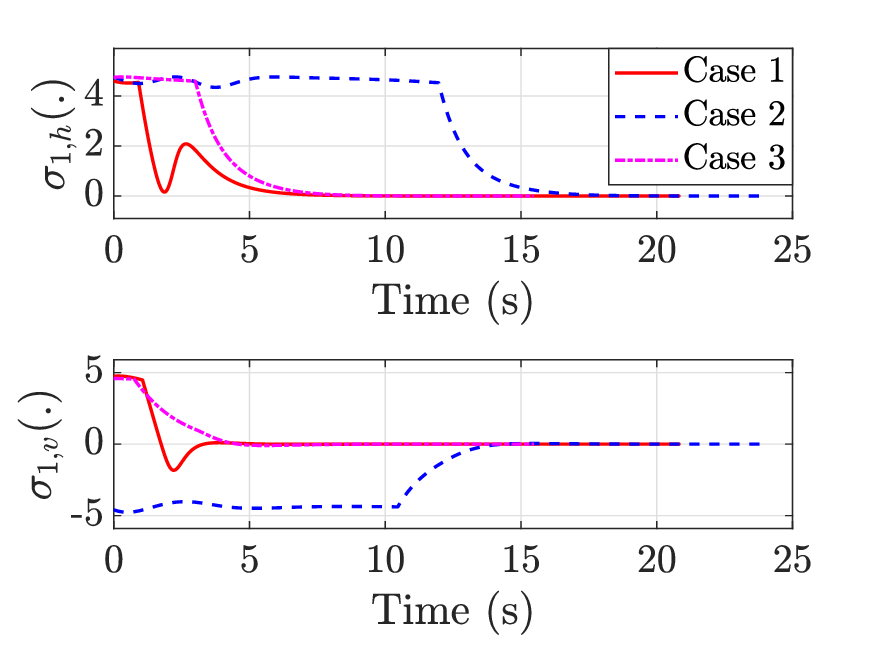}
			\caption{Saturation Functions.}
			\label{fig:11e}    
		\end{subfigure}%
		\begin{subfigure}{0.33\linewidth}
			\centering
			\includegraphics[width=\linewidth]{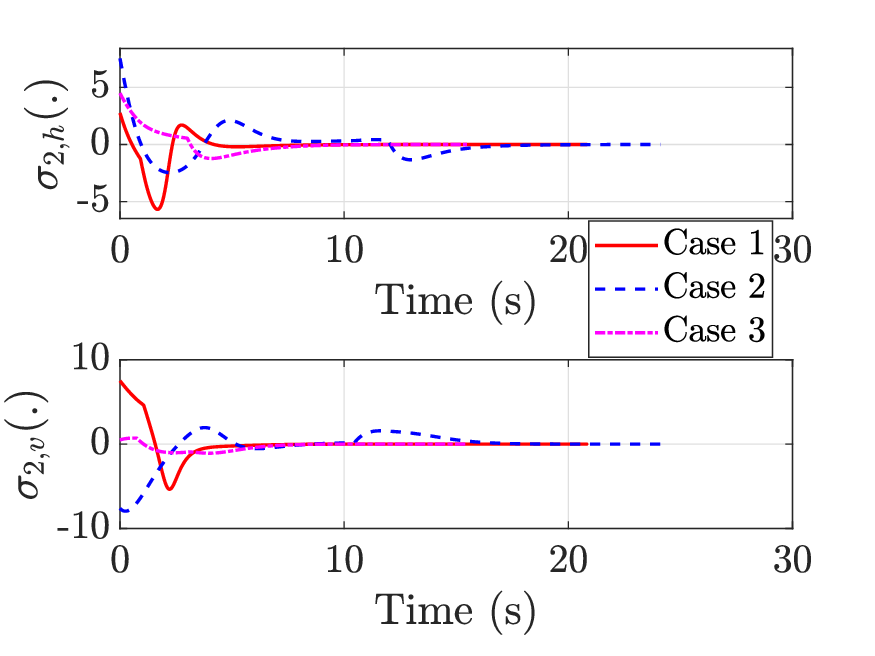}
			\caption{Saturation Functions.}
			\label{fig:11f}    
		\end{subfigure}
		\caption{UAV performance: 3D straight-line path}
		\label{fig:11}
	\end{figure*}
	
	Next, we consider the 3D circular path of radius 100 m and with its centre located at $(20\text{m},30\text{m},40\text{m})$. The roll angle and azimuthal angle of the circle are $20^\circ$ and $30^\circ$, respectively. The simulation parameters are chosen as follow: UAV speed of $15$ m/s, maximum acceleration bounds $a_{m\max} = M_{2,h}'=M_{2,v}'=15$ m/s$^2$, and gains chosen as $k_{1,h}=k_{2,h}=k_{1,v}=k_{2,v}=5$. To evaluate the performance, three distinct initial conditions are chosen: $(x,y,z,\chi,\gamma) = (70, 70, 90, 30^\circ, 40^\circ)$, $(-80, -80, 100, 60^\circ, 30^\circ)$, and $(150, 80, 30, 40^\circ, 40^\circ)$. Under the proposed guidance law, the performance of the UAV is demonstrated through \Cref{fig:12}. As observed, the UAV converges to the desired circular path in all cases, driving the horizontal and vertical states $d_h, d_v, \dot{d}_h, \dot{d}_v$ to zero (see \Cref{fig:11c}, \Cref{fig:11d}). The control inputs $a_h$ and $a_v$ remain strictly bounded within the prescribed limits of $\pm 10$ m/s$^2$ for the entire duration $t \geq 0$ (see \Cref{fig:12b}). A similar response is observed in the circular path following case as in the straight-line scenario.
	\begin{figure*}[htbp]
		\centering
		\begin{subfigure}{0.33\linewidth}
			\centering
			\includegraphics[width=\linewidth]{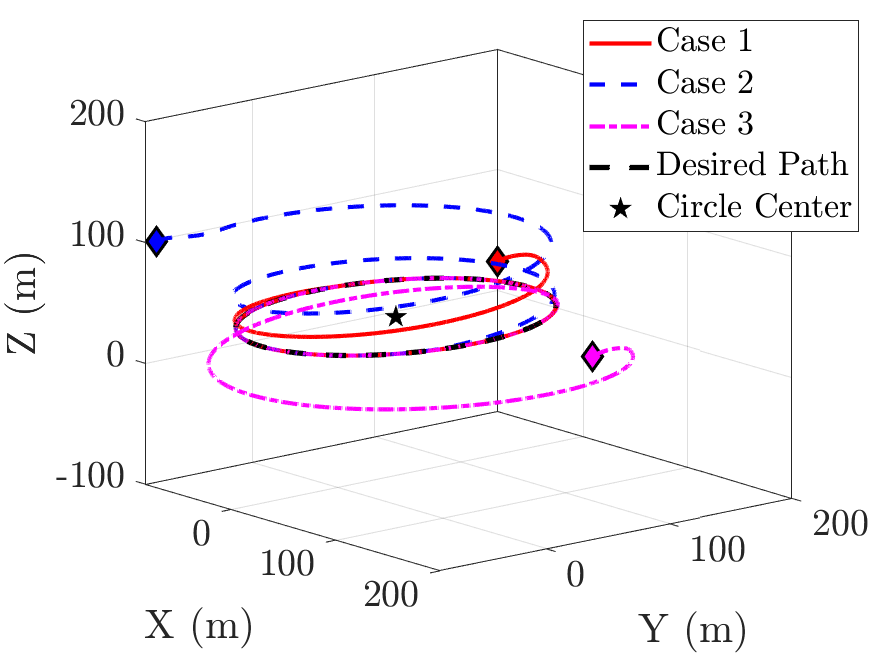}
			\caption{Trajectories.}
			\label{fig:12a}    
		\end{subfigure}%
		\begin{subfigure}{0.33\linewidth}
			\centering
			\includegraphics[width=\linewidth]{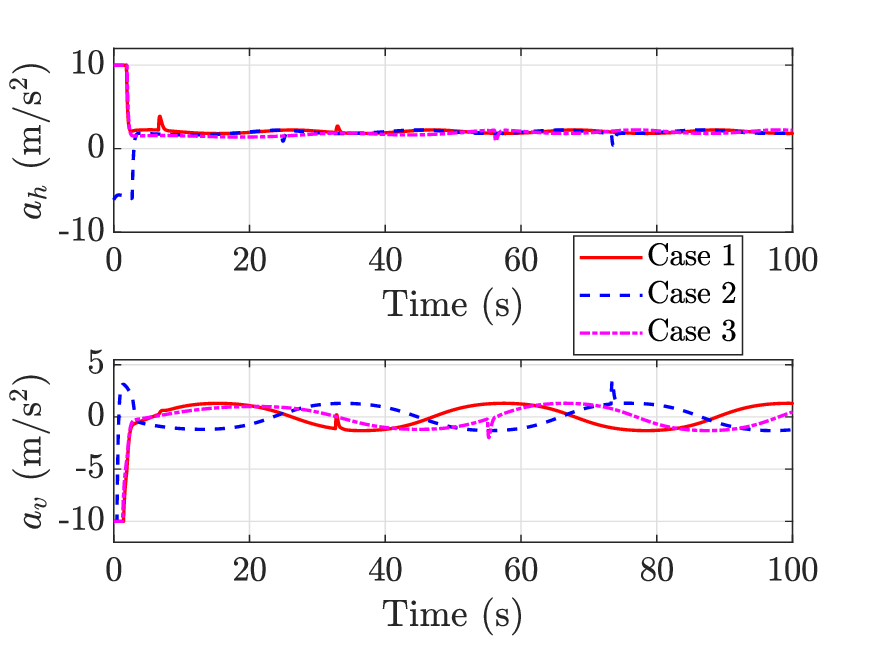}
			\caption{Lateral accelerations.}
			\label{fig:12b}    
		\end{subfigure}%
		\begin{subfigure}{0.33\linewidth}
			\centering
			\includegraphics[width=\linewidth]{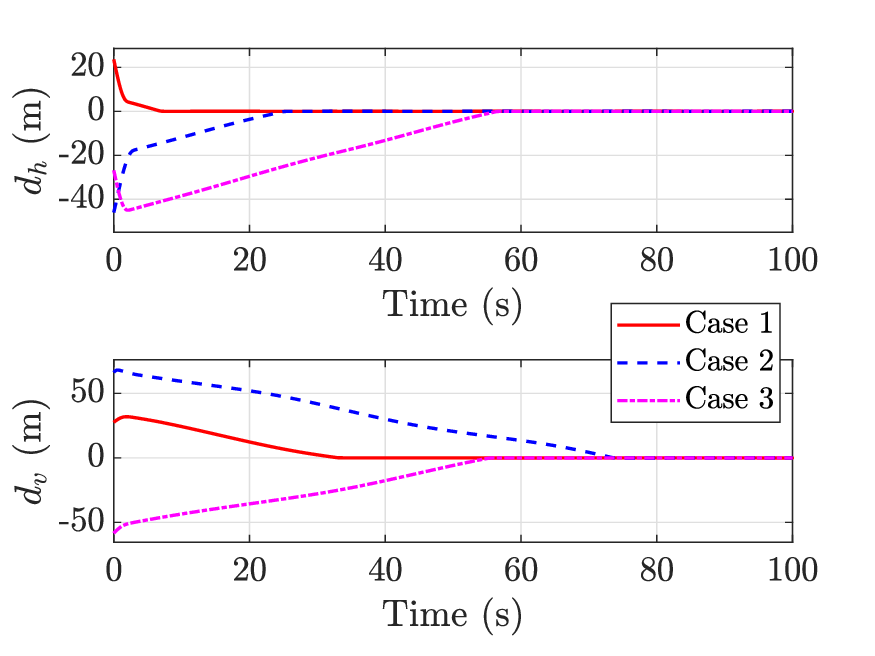}
			\caption{Distance Error.}
			\label{fig:12c}
		\end{subfigure}
		\begin{subfigure}{0.33\linewidth}
			\centering
			\includegraphics[width=\linewidth]{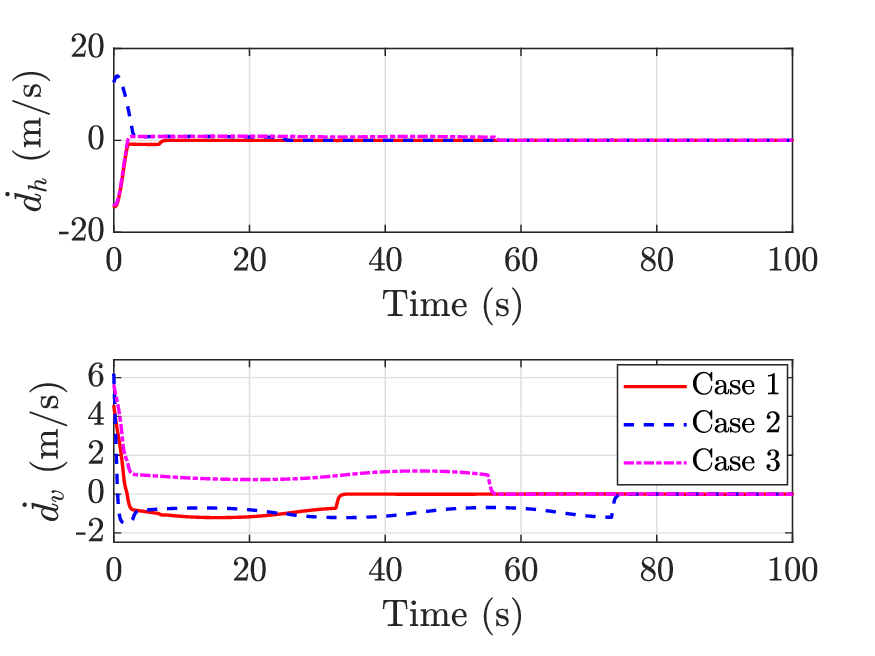}
			\caption{Velocity Error.}
			\label{fig:12d}    
		\end{subfigure}%
		\begin{subfigure}{0.33\linewidth}
			\centering
			\includegraphics[width=\linewidth]{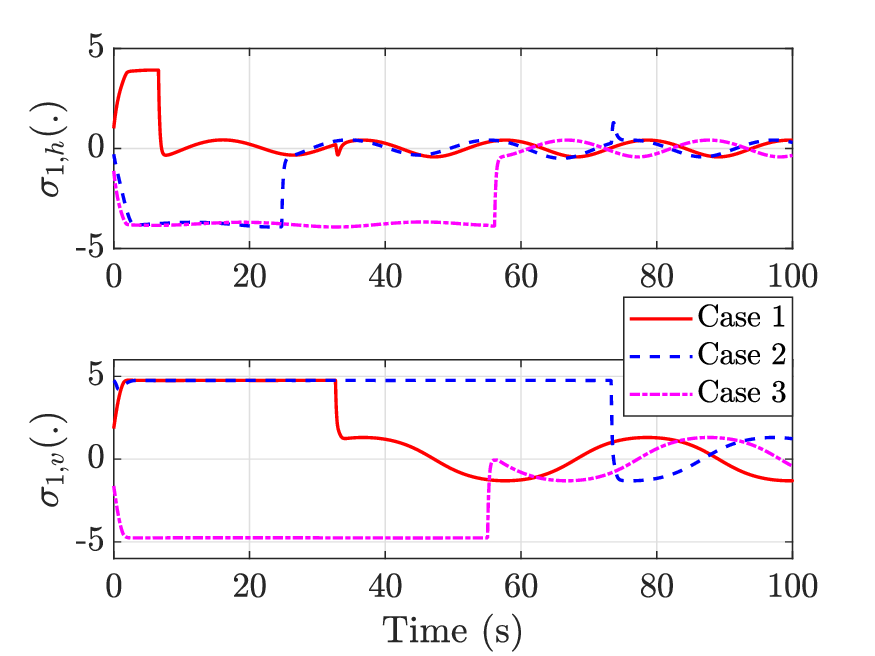}
			\caption{Saturation Function.}
			\label{fig:12e}    
		\end{subfigure}%
		\begin{subfigure}{0.33\linewidth}
			\centering
			\includegraphics[width=\linewidth]{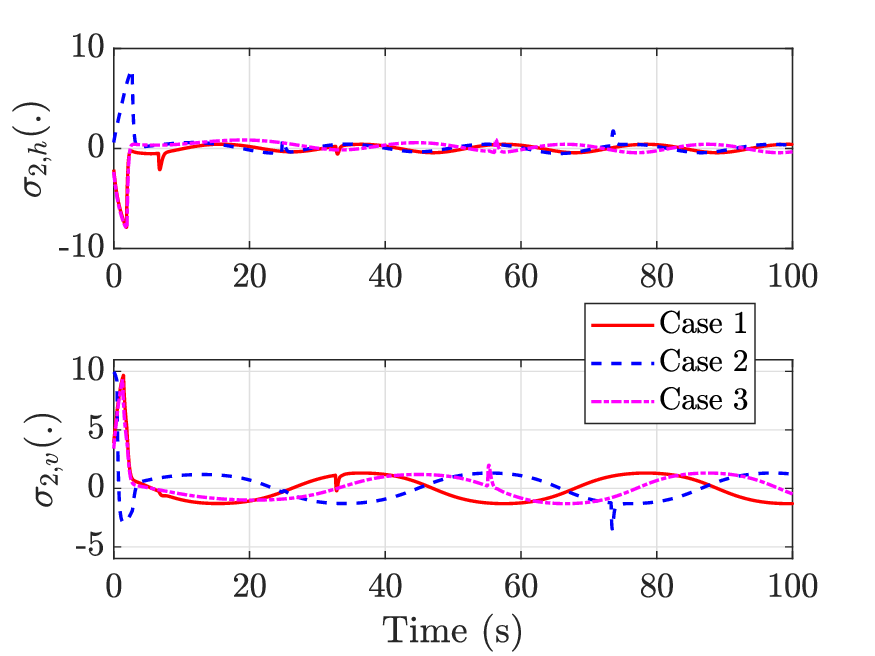}
			\caption{Saturation Function.}
			\label{fig:12f}    
		\end{subfigure}
		\caption{UAV performance: 3D circular path.} 
		\label{fig:12}
	\end{figure*}
	
	We now demonstrate the robustness of the proposed guidance law, when the UAV is following a 3D path. For this, we consider a scenario in which a UAV has to follow a 3D circular path in the presence of a wind gust. The simulation parameters are chosen the same as the 3D circular path following. The initial conditions are chosen as $x=y=z=100$ m, and $\chi=\gamma=45^\circ$. The wind gust is characterized by a velocity of $5\sqrt{3}$ m/s with an azimuthal angle of $45^\circ$ and an elevation angle of $ 35.26^\circ$, and it is active for 10 seconds, as shown in \Cref{fig:13e}. The performance of the UAV under the proposed guidance law during the wind disturbance is shown in \Cref{fig:13}. UAV successfully converges to the desired 3D circular path in the presence of  wind gust. This, in turn, bolsters to the claim that the proposed guidance strategy is robust against such external disturbances.
	\begin{figure*}[!ht]
		\centering
		\begin{subfigure}{0.5\linewidth}
			\centering
			\includegraphics[width=\linewidth]{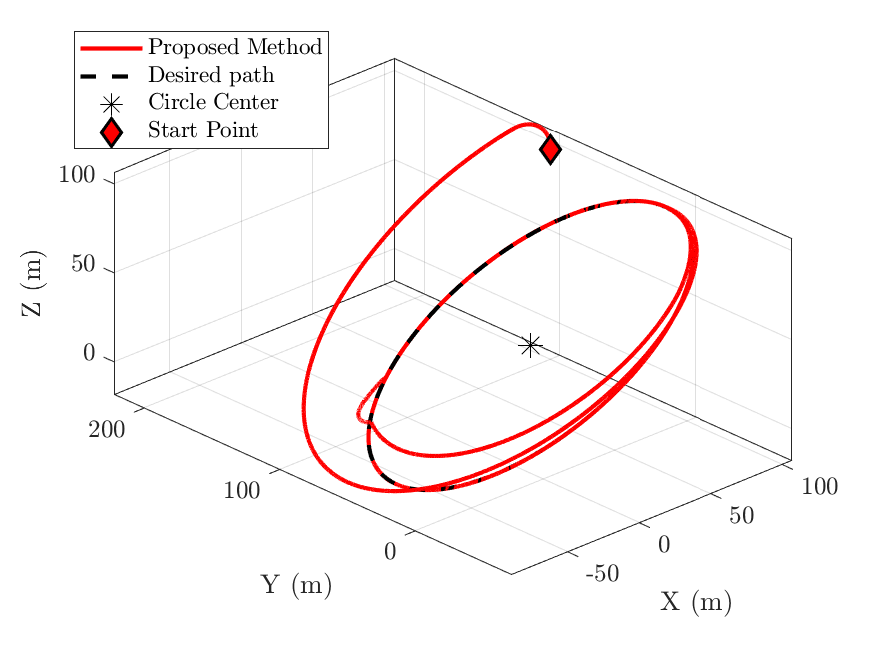}
			\caption{Trajectories.}
			\label{fig:13a}    
		\end{subfigure}%
		\begin{subfigure}{0.5\linewidth}
			\centering
			\includegraphics[width=\linewidth]{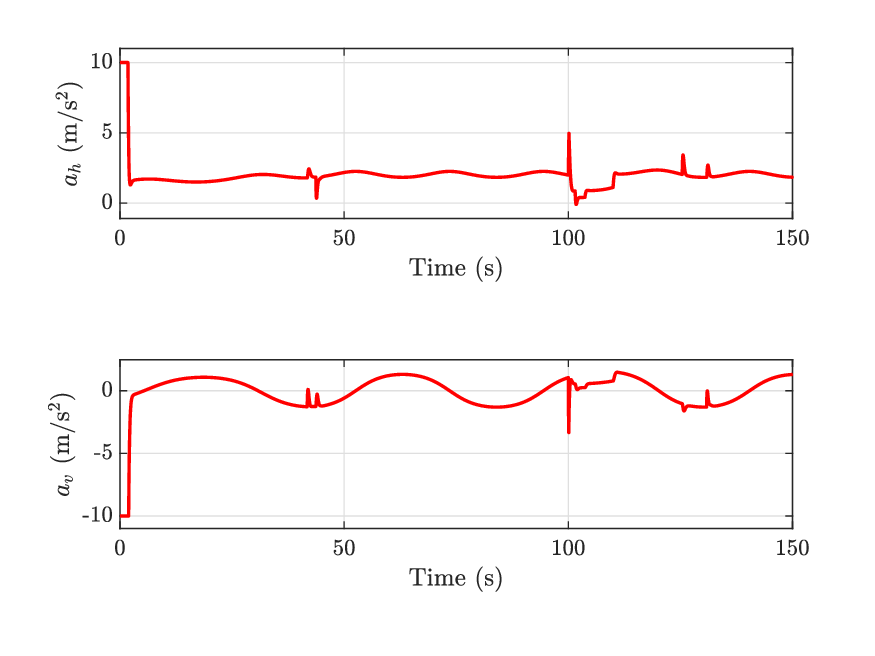}
			\caption{Lateral accelerations.}
			\label{fig:13b}    
		\end{subfigure}
		\begin{subfigure}{0.33\linewidth}
			\centering
			\includegraphics[width=\linewidth]{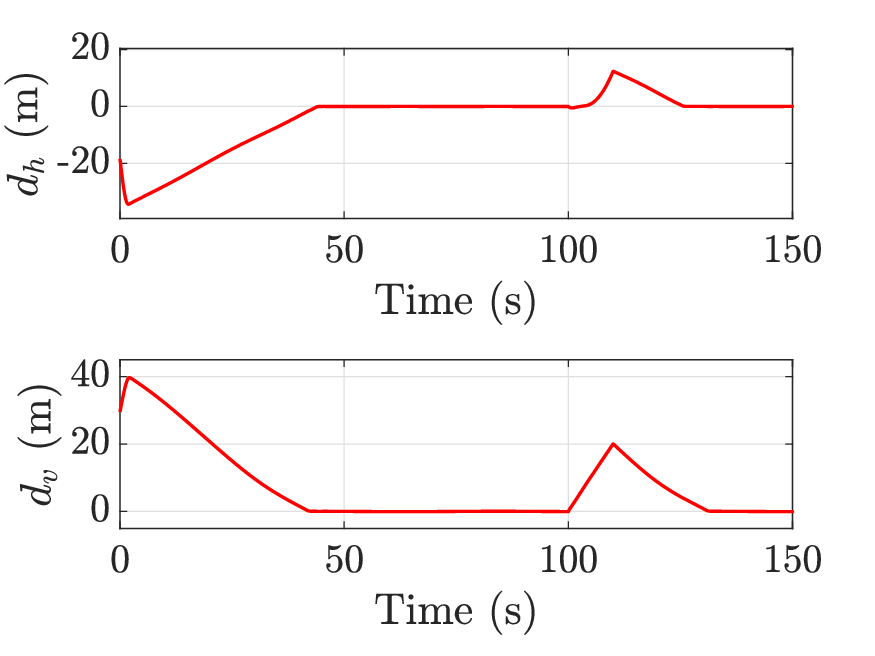}
			\caption{Distance Error.}
			\label{fig:13c}    
		\end{subfigure}%
		\begin{subfigure}{0.33\linewidth}
			\centering
			\includegraphics[width=\linewidth]{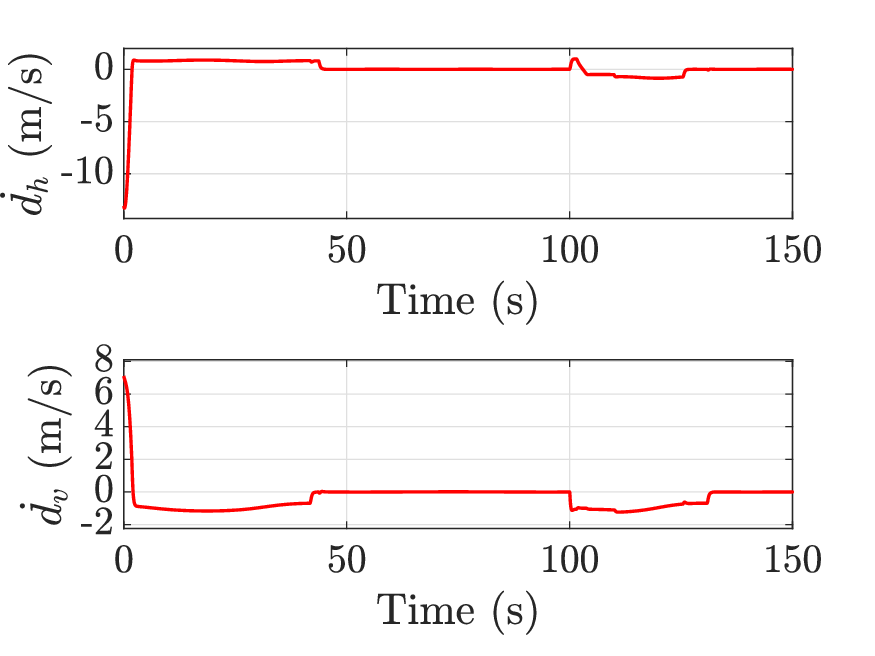}
			\caption{Velocity Error.}
			\label{fig:13d}    
		\end{subfigure}%
		\begin{subfigure}{0.33\linewidth}
			\centering
			\includegraphics[width=\linewidth]{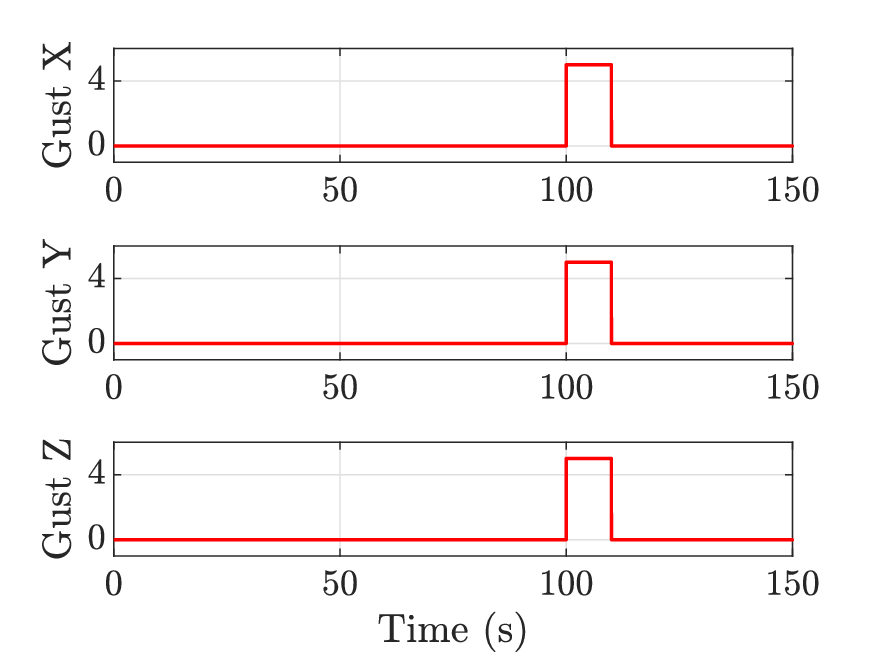}
			\caption{Gust.}
			\label{fig:13e}    
		\end{subfigure}
		\caption{UAV's performance: in the presence of gust.}
		\label{fig:13}
	\end{figure*}
	\section{Conclusions}\label{sec:conclusion}
	In this paper, we proposed a bounded input path following guidance law for UAV, based on nested saturation theory. The proposed guidance law ensures that the control inputs remain strictly within predefined bounds for all initial conditions, making it suitable for real-world applications with actuator constraints. We demonstrate the performance of the proposed controller across various paths, including straight lines, circular orbits, and sinusoidal curves, as well as in the presence of wind, which confirms the robustness of the proposed strategy. The results demonstrate that the proposed guidance law steers the UAV to the desired path with smooth convergence and bounded control input. We presented a comparative analysis with existing path-following strategies, which further highlighted the advantages of the proposed method, particularly in terms of input feasibility, tracking performance, and energy efficiency. The reduced control effort enhances the battery life, which is very critical for long-endurance missions. The proposed controller is simple and easy to implement on an inexpensive on-board sensor, making it suitable for real-world applications such as aerial surveillance, environmental monitoring, and maneuvering. The future direction could be  the design of guidance law accounting for various kinds of obstacles under input constraints.
	\bibliography{references_1.bib} 
\end{document}